\documentclass[journal,onecolumn]{IEEEtran}

\newcommand{\no}[1]{}
\newcommand{\dd}{\mathinner{.\,.}}

\newcommand{\Sf}{\mathcal{S}}
\newcommand{\eps}{\varepsilon}

\newcommand{\Symb}{\mathcal{A}}
\newcommand{\Act}{\mathcal{B}}
\newcommand{\Pres}{\mathcal{S}}
\newcommand{\Left}{\mathcal{L}}
\newcommand{\Right}{\mathcal{R}}
\newcommand{\rle}{\mathsf{rle}}
\newcommand{\pc}{\mathsf{pc}}
\newcommand{\Zz}{\mathbb{Z}_{\ge 0}}
\newcommand{\Zp}{\mathbb{Z}_{+}}
\newcommand{\val}{\exp}
\newcommand{\sub}{\subseteq}
\newcommand{\sm}{\setminus}
\newcommand{\ceil}[1]{\left\lceil #1 \right\rceil}
\newcommand{\floor}[1]{\left\lfloor #1 \right\rfloor}
\newcommand{\Exp}{\mathbb{E}}
\newcommand{\Oh}{O}
\newcommand{\id}{\mathsf{id}}

\newcommand{\per}{\mathsf{per}}

\usepackage{graphicx}
\usepackage{color}
\usepackage{amsmath}
\usepackage{amsthm}
\usepackage{amssymb}

\usepackage{thm-restate}

\usepackage[hidelinks]{hyperref}
\usepackage[capitalize]{cleveref}

\newtheorem{theorem}{Theorem}[section]

\newtheorem{definition}[theorem]{Definition}
\newtheorem{lemma}[theorem]{Lemma}
\newtheorem{proposition}[theorem]{Proposition}

\newtheorem{corollary}[theorem]{Corollary}
\newtheorem{claim}[theorem]{Claim}
\newtheorem{construction}[theorem]{Construction}

\begin{document}


\title{Towards a Definitive Compressibility Measure \\ for Repetitive Sequences}

\author{
Tomasz Kociumaka,
%
Gonzalo Navarro,
%
and Nicola Prezza
\thanks{A previous partial version of this article appeared in \emph{Proc. LATIN 2020}~\cite{KNP20}.}
\thanks{Tomasz Kociumaka is with
IEOR Dept., University of California, Berkeley, U.S.
Gonzalo Navarro is with Millennium Institute for Foundational Research on Data (IMFD), Dept.\ of Computer Science, University of Chile, Chile.
Nicola Prezza is with Ca’ Foscari University of Venice, Italy.}
}

\maketitle

\begin{abstract}
Unlike in statistical compression, where Shannon's entropy is a definitive lower bound, no such clear measure exists for the compressibility of repetitive sequences. Since statistical entropy does not capture repetitiveness, ad-hoc measures like the size $z$ of the Lempel--Ziv parse are frequently used to estimate it.
The size $b \le z$ of the smallest bidirectional macro scheme captures better what can be achieved via copy-paste processes, though it is NP-complete to compute and it is not monotonic upon symbol appends.
Recently, a more principled measure, the size $\gamma$ of the smallest string \emph{attractor}, was introduced.
The measure $\gamma \le b$ lower bounds all the previous relevant ones, yet length-$n$ strings can be represented and efficiently indexed within space $O(\gamma\log\frac{n}{\gamma})$, which also upper bounds most measures.
While $\gamma$ is certainly a better measure of repetitiveness than $b$, it is also NP-complete to compute and not monotonic, and it is unknown if one can always represent a string in $o(\gamma\log n)$ space.

In this paper, we study an even smaller measure, $\delta \le \gamma$, which can be computed in linear time, is monotonic, and allows encoding every string in $O(\delta\log\frac{n}{\delta})$ space because $z = O(\delta\log\frac{n}{\delta})$. We show that $\delta$ better captures the compressibility of repetitive strings. Concretely, we show that (1) $\delta$ can be strictly smaller than $\gamma$, by up to a logarithmic factor; (2) there are string families needing $\Omega(\delta\log\frac{n}{\delta})$ space to be encoded, so this space is optimal for every $n$ and $\delta$; (3) one can build run-length context-free grammars of size $O(\delta\log\frac{n}{\delta})$, whereas the smallest (non-run-length) grammar can be up to $\Theta(\log n/\log\log n)$ times larger; and (4) within $O(\delta\log\frac{n}{\delta})$ space we can not only represent a string, but also offer logarithmic time access to its symbols and efficient indexed searches for pattern occurrences.
\end{abstract}




\section{Introduction}

The recent rise in the amount of data we aim to handle is driving research into compressed data representations that can be used directly in compressed form~\cite{Nav16}. Interestingly, much of today's fastest-growing data is highly repetitive: genome collections, versioned text and software repositories, periodic sky surveys, and other sources produce data where each element in the collection is very similar to others.

Since a significant fraction of the data of interest consists of sequences, compression of highly repetitive text collections is gaining attention, as it enables space reductions of orders of magnitude~\cite{Nav20}. Statistical compression, however, is unable to capture repetitiveness~\cite{KN13}. Large space reductions are instead achieved with other kinds of compressors, such as Lempel--Ziv~\cite{LZ76}, grammar compression~\cite{KY00}, and the run-length-compressed Burrows--Wheeler transform~\cite{MNSV09}.

A fundamental question is \emph{how much compression can be achieved on repetitive collections}, or alternatively, \emph{how to measure data (compressibility by exploiting) repetitiveness}.
Unlike statistical compression, where Shannon's notion of entropy~\cite{Sha48} is a clear lower bound to what compressors can achieve, a similar notion capturing repetitiveness has been elusive. Beyond Kolmogorov's complexity~\cite{Kol65}, which is uncomputable, repetitiveness is measured in ad-hoc terms, as the result of what specific compressors achieve. A list of such measures on a string $S[1\dd n]$ includes:

\begin{LaTeXdescription}
\item[Lempel--Ziv compression~\cite{LZ76}] parses $S$ into a sequence of \emph{phrases}, each phrase being the longest string that occurs starting to the left in $S$. The associated measure is the number $z$ of phrases produced, which can be computed in $O(n)$ time~\cite{RPE81}.
\item[Bidirectional macro schemes~\cite{SS82}] extend Lempel--Ziv so that the source of each phrase may precede or follow it, as long as no circular dependencies are introduced. 
The associated measure $b$ is the number of phrases of the smallest parsing. It satisfies $b \le z = O(b\log\frac{n}{b})$~\cite{NOP20}, but computing $b$ is NP-complete~\cite{Gal82}.
\item[Grammar-based compression~\cite{KY00}] builds a context-free grammar that generates (only) $S$. 
The associated measure is the size $g$ of the smallest grammar (i.e., the total length of the right-hand sides of the rules). 
It satisfies $z \le g = O(z\log\frac{n}{z})$ and, while it is NP-complete to compute $g$, grammars of size $O(z\log\frac{n}{z})$ can be constructed in linear time~\cite{Ryt03,CLLPPSS05,Jez16}.
\item[Run-length grammar compression~\cite{NIIBT16}] allows in addition rules $A \rightarrow B^t$ ($t$ repetitions of $B$) of constant size. The measure is the size $g_{rl}$ of the smallest run-length grammar, and it satisfies $z \le g_{rl} \le g$ and $g_{rl}=O(b\log\frac{n}{b})$~\cite{NOP20}.
\item[Collage systems~\cite{KMSTSA03}] extend run-length grammars by allowing truncation: in constant space, we can refer to a prefix or a suffix of another nonterminal. The associated measure $c$ satisfies $c \le g_{rl}$, $b = O(c)$, and $c=O(z)$~\cite{NOP20}.
\item[Burrows--Wheeler transform (BWT)~\cite{BW94}]
is a permutation of $S$ that tends to have long runs of equal letters if $S$ is repetitive. The number $r$ of maximal equal-letter runs in the BWT can be found in linear time. It is known that $\frac{b}{2} \le r=O(b\log^2 n)$~\cite{NOP20,KK19}.
\item[CDAWGs~\cite{BBHMCE87}] are automata that recognize every substring of $S$.
The associated measure of repetitiveness is $e$, the size of the smallest such automaton (compressed by dissolving states of in-degree and out-degree one), which can be built in linear time~\cite{BBHMCE87}.
The measure $e$ is always larger than $r$, $g$, and $z$~\cite{BCGPR15,BC17}.
\item[Lex parsing~\cite{NOP20}] is analogous to Lempel--Ziv parsing, but each phrase must point to a lexicographically smaller source. The lex parsing is computed in linear time. Its number of phrases, $v$, satisfies $\frac{b}{2} \le v \le 2r$ and $v \le g_{rl}$~\cite{NOP20}.
\end{LaTeXdescription}

For each measure $x$ above, we can represent $S[1\dd n]$ in space $O(x)$ (meaning $O(x\log n)$ bits in this article). As seen, the measures form a complex hierarchy of dominance relations~\cite{Nav20}, where $b$ asymptotically dominates all the others. A problem with $b$ (and also $z$, $c$, $r$, and $v$) is that it is unknown how to \emph{access} $S$ (i.e., extract any character $S[i]$) efficiently (say, in $n^{o(1)}$ time, i.e., without decompressing much of it) within space $O(b)$ (or $O(\max(z,c,r,v))$). This has been achieved in time $O(\log n)$, but only within space $O(z\log\frac{n}{z})$~\cite{CLLPPSS05,Ryt03}, $O(e)$~\cite{BCspire17}, $O(r\log\frac{n}{r})$~\cite{GNP18}, $O(g)$~\cite{BLRSRW15}, and even $O(g_{rl})$~\cite{CEKNP19}, the latter of which is $O(b\log\frac{n}{b})$ and subsumes all the other spaces. Providing direct access to the sequences is essential for manipulating them in compressed form, without ever having to decompress them.

Just accessing the string is not sufficient, however, for many applications. One of the most fundamental text processing tasks is \emph{string matching}: find all the occurrences in $S$ of a short string $P$. This is particularly challenging when the string $S$ is large and scanning it sequentially is not viable. We then resort to \emph{indexes}, which are data structures offering $n^{o(1)}$-time string matching (and possibly other more sophisticated capabilities) over a collection of strings.
Statistically compressed text indexes are already mature~\cite{NM07} but, as explained earlier, are insensitive to repetitiveness. 
Various more recent compressed indexes build on the repetitiveness measures above; see a thorough review~\cite{Nav20.2}.
The smallest of those find the $occ$ occurrences of $P[1\dd m]$ in $O(g_{rl})$ space and $O(m\log n + occ \log^\epsilon n)$ time, for any constant $\epsilon>0$ \cite{CEKNP19}, or $O(r)$ space and $O(m\log\log\sigma+occ)$ time over an alphabet of size $\sigma$~\cite{GNP18,NT20}.
Just counting the number of occurrences can be done in space $O(g)$ (not $O(g_{rl})$) and time $O(m\log^{2+\epsilon} n)$ \cite{CEKNP19}, or space $O(r)$ and time $O(m\log\log\sigma)$ \cite{GNP18,NT20}. 

A relevant recent development in measuring repetitiveness is the concept of string \emph{attractor}~\cite{KP18}.
An attractor $\Gamma$ is a set of positions in $S$ such that any substring of $S$ has an occurrence covering a position in $\Gamma$. Since $\gamma = O(b)$ \cite{KP18}, the size of the smallest attractor asymptotically lower bounds all the repetitiveness measures listed above; however, it is unknown if one can represent any string in $O(\gamma)$ space. We can in space $O(\gamma\log\frac{n}{\gamma})$, within which we can also access any symbol of $S$ in time $O(\log \frac{n}{\gamma})$ \cite{KP18}, and even support indexed text searching \cite{NP18} within time as low as $O(m + (occ+1)\log^\epsilon n)$ for locating all the occurrences and $O(m+\log^{2+\epsilon} n)$ time for counting them \cite{CEKNP19}. It is known that $g_{rl}=O(\gamma\log\frac{n}{\gamma})$~\cite{CEKNP19}, though $g_{rl}$ can be smaller than $\gamma\log\frac{n}{\gamma}$ by up to a logarithmic factor, $\log\frac{n}{\gamma}$, so the slower index of size $O(g_{rl})$ offers better space in general.

In terms of measuring repetitiveness, both $\gamma$ and $b$ share some unsatisfactory aspects. Both are NP-hard to compute~\cite{Gal82,KP18}, and both are non-monotone when $S$ grows by appending characters at the endpoints of the string~\cite{Nav20,MRRRS20}.

\subsection{Our contributions}

In this paper, we study a new measure of repetitiveness, $\delta$, which arguably captures better the concept of compressibility in repetitive strings and is more convenient to deal with.
Although this measure was already introduced in a stringology context~\cite{RRRS13} and used to build indexes of size $O(\gamma\log\frac{n}{\gamma})$ without knowing $\gamma$~\cite{CEKNP19}, its properties and full potential have not been explored.
It is known that $\delta \le \gamma$ for every string, that $\delta$ can be computed in $O(n)$ time~\cite{CEKNP19}, and that $z=O(\delta\log\frac{n}{\delta})$~\cite{RRRS13}. Further, $\delta$ is insensitive to string reversals and alphabet permutations, and monotone upon appending symbol, unlike $\gamma$, $b$, or~$z$~\cite{Nav20,MRRRS20}. We prove several further properties related to $\delta$:

\begin{enumerate}
\item In Section~\ref{sec:lb-attr}, we show that $\delta$ can be strictly smaller than $\gamma$, by up to a logarithmic factor. More precisely, for any $n$ and $\delta$, there are strings with $\gamma=\Omega(\delta\log\frac{n}{\delta})$. We therefore show that the already known upper bounds $\gamma, b, c, z = O(\delta\log\frac{n}{\delta})$ are tight for every $n$ and $\delta$. 
\item In Section~\ref{sec:lb-entr}, we show that $O(\delta\log\frac{n}{\delta}\log n)$ bits, a space one can reach by Lempel--Ziv compression due to $z = O(\delta\log\frac{n}{\delta})$, is indeed tight: for every $n$ and $\delta$, there are string families that need $\Omega(\delta\log\frac{n}{\delta}\log n)$ bits to be represented. Instead, the upper bound $O(\gamma\log\frac{n}{\gamma})$~\cite{KP18} is not known to be tight.
\item In Section~\ref{sec:ub-gram}, we show that not only the Lempel--Ziv parsing, but also run-length context-free grammars, can always represent a string within $O(\delta\log\frac{n}{\delta})$ space, thus also $v,g_{rl} = O(\delta\log\frac{n}{\delta})$. However, standard context-free grammars cannot: for every $n$ and $\delta$, there strings satisfying $g = \Omega(\delta\log^2\frac{n}{\delta}/\log\log\frac{n}{\delta})$. In particular, if $\delta = n^{1-\Omega(1)}$, this lower bound simplifies to $g = \Omega(\delta\log^2 n /\log\log n)$, which is almost a logarithmic factor away from $\delta\log\frac{n}{\delta}=\Theta(\delta \log n)$.
\item In Section~\ref{sec:bt}, we combine our preceding result with previous ones on run-length grammars~\cite{CEKNP19} to show that, within space $O(\delta\log\frac{n}{\delta})$, we can not only represent a string but also provide access to any position of it in time $O(\log n)$, compute substring fingerprints in time $O(\log n)$, find the $occ$ occurrences of any pattern string $P[1\dd m]$ in time $O(m\log n + occ \log^\epsilon n)$ for any constant $\epsilon>0$, and count them in time $O(m\log^{2+\epsilon} n)$. 
Furthermore, we show that the block tree data structure~\cite{BGGKOPT15}, which provides access to string symbols and substring fingerprints in time $O(\log\frac{n}{z})$, is of size $O(\delta\log\frac{n}{\delta})$, improving upon our first result and on previous analyses and variants of block trees~\cite{KP18,NP18,prezza2019optimal}.
\end{enumerate}


\section{Basic concepts and the measure $\delta$}\label{sec:delta}

We consider strings $S[1\dd n]$ as sequences of $|S|=n$ symbols $S[1],\ldots,S[n]$, each drawn from an alphabet $\Sigma=\{1,\ldots,\sigma\}$. For simplicity, we assume that every symbol of $\Sigma$ appears in $S$, though our results hold as long as $\sigma=n^{O(1)}$. The concatenation of strings $S$ and $S'$ is denoted $S \cdot S'$; we can also identify individual symbols of $\Sigma$ with the corresponding string of length $1$. A substring of $S$ is denoted $S[i\dd j] = S[i] \cdots S[j]$ and the empty string is denoted $\varepsilon$.

We assume the transdichotomous RAM model, which is a word RAM model on a machine word of $\Theta(\log n)$ bits. Consequently, when we measure the space in words (the default), we consider that each word holds $\Theta(\log n)$ bits. Therefore $O(x)$ space is equivalent to $O(x\log n)$ bits of space.

The measure $\delta$ was recently defined by Christiansen et al.~\cite[Sec.~5.1]{CEKNP19}, though it is based on the expression $d_k(S)/k$, introduced by Raskhodnikova et al.~\cite{RRRS13} to approximate $z$. The set of values $d_k(S)$ are known as the substring complexity of $S$, so $\delta$ is a function of it.
In this section, we summarize what is known about $\delta$.

\begin{definition}\label{def:delta}
Let $d_k(S)$ be the number of distinct length-$k$ substrings in 
$S$.~Then \[\delta = \max \{ d_k(S)/k : k\in [1\dd n]\}.\]
\end{definition}

\begin{lemma}[cf.~{\cite[Lemma~3]{RRRS13}}]\label{lem:z}
It always holds that $z = O(\delta\log\frac{n}{\delta})$.
\end{lemma}
\begin{proof}
Raskhodnikova et al.~\cite{RRRS13} proved that if $d_{\ell}(S) \le m\cdot \ell$ for every $\ell\le \ell_0$,
then $z \le 4(m \log \ell_0  + \frac{n}{\ell_0})$. Plugging $\ell_0 = \frac{n}{\delta}$ and $m=\delta$,
we conclude that $z \le 4(\delta \log\frac{n}{\delta} + \delta)=O(\delta \log \frac{n}{\delta})$.
\end{proof}

Since $\gamma$, $b$, and $c$ are $O(z)$, these three measures are all upper bounded by $O(\delta \log \frac{n}{\delta})$.
Additionally, we conclude that $g_{rl}\le g = O(z \log \frac{n}{z})=O(\delta \log^2 \frac{n}{\delta})$,
and note that $r = O(\delta \log \delta \max(1,\log\frac{n}{\delta\log\delta}))$ has been proved recently~\cite{KK19};
the latter bound is also known to be tight for any $\delta$ between $\Omega(1)$ and $O(n)$.

Before we proceed, let us recall the concept of an attractor.

\begin{definition}[Kempa and Prezza~{\cite{KP18}}]
An \emph{attractor} of a string
$S[1\dd n]$ is a set of positions $\Gamma \subseteq [1\dd n]$ such that every substring
$S[i\dd j]$ has an occurrence $S[i'\dd j']=S[i\dd j]$ that covers an attractor
position $p\in \Gamma \cap [i'\dd j']$. 
\end{definition}

\begin{lemma}[{\cite[Lemma~5.6]{CEKNP19}}]\label{lem:delta}
    Every string $S$ satisfies $\delta \le \gamma$.
    \end{lemma}
    \begin{proof}
    Every length-$k$ substring 
    has an occurrence covering
    an attractor position, so there can be at most $k\gamma$ 
    distinct substrings of length $k$, that is, $d_k(S)/k \le \gamma$ for all $k\leq n$.
    \end{proof}

\begin{lemma}[{\cite[Lemma~5.7]{CEKNP19}}]\label{lem:compute delta}
The measure $\delta$ can be computed in $O(n)$ time and space given $S[1\dd n]$.
\end{lemma}
\begin{proof}
One can use the suffix tree or the LCP table of $S$ to retrieve $d_k(S)$ for all $k\in [1\dd n]$ in $O(n)$ time, and
then compute $\delta$ from this information.
\end{proof}

Finally, we note some obvious positive properties of $\delta$ as a compressibility measure: it is insensitive to reversing the string and to alphabet permutations, and it is monotone when we add/remove prefixes/suffixes to/from $S$. Other measures, like $z$, $b$, or $\gamma$, are not monotone, $z$ is sensitive to reversals, and $v$ and $r$ are also sensitive to alphabet permutations~\cite{Nav20,MRRRS20}.

\section{Lower bounds on attractors}\label{sec:lb-attr}

In this section,
we show that there exist string families where $\delta=o(\gamma)$; in fact, $\delta$ can be smaller by up to a logarithmic factor. More precisely, for any string length $n$ and value $\delta \in [2\dd n]$, we build a string satisfying $\gamma = \Omega(\delta\log\frac{n}{\delta})$. This shows that the bound $\gamma \le b \le z = O(\delta\log\frac{n}{\delta})$ is asymptotically tight.

We build our results on (variants of) the following string family.

\begin{definition}\label{def:S}
Consider an infinite string $S_\infty [1\dd ]$, where $S_\infty [i] = \mathtt{b}$ if $i=2^j$ for some integer $j\ge 0$, and $S_\infty[i]=\mathtt{a}$ otherwise.
For $n\ge 1$, let $S_n = S_\infty[1\dd n]$. We then define the string family $\Sf = \{ S_n : n \ge 1 \}$ as the set of the prefixes of $S_\infty$.
\end{definition}

We first prove that the strings in $\Sf$ satisfy $\delta=O(1)$ and $\gamma = \Omega(\log n)$.

\begin{lemma}\label{lem:Sn}
For every $n\ge 1$, the string $S_n$ satisfies $\delta \le 2$ and $\gamma \ge \frac12\lfloor\log n\rfloor$.
\end{lemma}
\begin{proof}
For each $j\ge 1$, every pair of consecutive $\mathtt{b}$s in $S_\infty[2^{j-1}+1\dd]$ is at distance at least $2^j$. 
Therefore, the only distinct substrings of length $k\le 2^j$ in $S_\infty[2^{j-1}+1\dd]$ are of the form $\mathtt{a}^k$ or $\mathtt{a}^i \mathtt{b} \mathtt{a}^{k-i-1}$ for $i\in[0\dd k)$.
Hence, the distinct length-$k$ substrings of $S_\infty$ are those starting up to position $2^{j-1}$, $S_\infty[i\dd i+k)$ for $i\in [1\dd 2^{j-1}]$, and the $k+1$ already mentioned strings, for a total of $d_{k}(S_\infty) \le 2^{j-1}+k+1$. Choosing $j=\lceil \log k \rceil$, we get $d_k(S_\infty) \le 2^{\lceil \log k \rceil-1}+k+1 \le 2^{\log k}+k = 2k$, yielding that $\delta(S_n)\le 2$ holds for every $n$.

Next, observe that, for each $j\ge 0$, the substring $\mathtt{b} \mathtt{a}^{2^j-1} \mathtt{b}$ has its unique occurrence in $S_\infty$ at $S_\infty[2^{j}\dd 2^{j+1}]$.
The covered regions are disjoint across \emph{even} integers $j$, so each one requires a distinct attractor position.
Consequently, $\gamma(S_{n})\ge \frac{k}{2}$ holds for all integers $n\ge 2^{k}$ and $k\ge 0$. Choosing $k = \lfloor{\log n}\rfloor$, we get
$\gamma(S_{n}) \ge \frac12\lfloor\log n\rfloor$.
\end{proof}

We can now show that, for every integer $2\le \delta = o(n)$, there are strings satisfying $\gamma = \omega(\delta)$; that is, $\gamma$ can be asymptotically larger than~$\delta$.

\begin{theorem}\label{thm:lbgamma}
For every length $n$ and integer value $\delta\in [2\dd n]$, there is a string $S[1\dd n]$ with $\gamma = \Omega(\delta \log \frac{n}{\delta})$.
\end{theorem}
\begin{proof}
Let us first fix an integer $m\ge 1$ such that $n\ge 4m-1$ and decompose $n-m+1 \ge 3m$ into $\sum_{i=1}^m n_i = n-m+1$ roughly equally (so that $n_i \ge 3$ and $n_i = \Omega(\frac{n}{m})$). We shall build a string $S$ over an alphabet consisting of $3m-1$ characters: $\mathtt{a}_i$ and $\mathtt{b}_i$ for $i\in [1\dd m]$
and delimiters $\$_i$ for $i\in [1\dd m)$.
For this, we take $S^{(i)}$ to be the string $S_{n_i}$ of Definition~\ref{def:S}, with the alphabet $\{\mathtt{a},\mathtt{b}\}$ replaced
by $\{\mathtt{a}_i,\mathtt{b}_i\}$, and we define $S = S^{(1)}\,\$_1\, S^{(2)}\,\$_2\cdots \$_{m-1}\,S^{(m)}$, which is of length $n$. 

Notice that, for each $k\in [1\dd n]$, we have $d_k(S)\le (m-1)k + \sum_{i=1}^m d_k(S^{(i)})$ because every substring contains $\$_i$
or is contained in $S^{(i)}$ for some $i$. Since $d_k(S^{(i)}) \le 2k$ by \cref{lem:Sn}, we have that $d_k(S) \le (3m-1)k$, and thus $\delta(S)\le 3m-1$.
In fact, $\delta(S)=3m-1$ because $d_1(S)=3m-1$.
Furthermore, $\gamma(S)\ge \sum_{i=1}^m \gamma(S^{(i)}) \ge \sum_{i=1}^m \frac12\lfloor{\log n_i}\rfloor = \Omega(m \log \frac{n}{m})=\Omega(\delta \log \frac{n}{\delta})$,
where the first inequality holds because the alphabets of $S^{(i)}$ are disjoint and the second is due to \cref{lem:Sn}.

This construction proves the theorem for $\delta = 3m-1$ and $n\ge 4m-1$.
If $\delta < \frac34n$ and $\delta \bmod 3 \ne 2$, we use $m = \lfloor\frac{\delta+1}{3}\rfloor$ and initially construct a string $S$ of length $n-(\delta+1)\bmod 3 \ge 4m-1$. Next, we append $(\delta+1)\bmod 3$ additional delimiters, which results in each of the measures
$\delta(S)$, $\gamma(S)$, and $n$ increased by $(\delta+1)\bmod 3$.
Finally, we note that if $\delta \ge \frac34n = \Omega(n)$, then the claim reduces to $\gamma = \Omega(\delta)$ and therefore follows directly from \cref{lem:delta}.
\end{proof}

\section{Lower bounds on text entropy}\label{sec:lb-entr}

In this section,
we prove that there are string families that cannot be encoded in $o(\delta\log n)$ space: for every length $n$ and every integer value $\delta \in [2\dd n]$, there is a string family whose elements require $\Omega(\delta\log\frac{n}{\delta})$ space, or $\Omega(\delta\log \frac{n}{\delta}\log n)$ bits, to be represented.
Therefore, using $O(\delta\log\frac{n}{\delta})$ space is worst-case optimal for every $\delta$. This is a reachable bound, because every string can be represented within $O(z) \subseteq O(\delta\log\frac{n}{\delta})$ space. 
In comparison, it is not known if the upper bound $O(\gamma\log\frac{n}{\gamma})$ to encode every string family~\cite{KP18} is also tight.

We consider a family of variants of the infinite string $S_{\infty}$ of Definition~\ref{def:S},
where the positions of $\mathtt{b}$s are further apart and slightly perturbed.

\begin{definition}\label{def:Sstar}
The family $\Sf^\mathtt{p}$ is formed by all the infinite strings $S$ over $\{ \mathtt{a},\mathtt{b}\}$ where the first $\mathtt{b}$ is placed at $S[1]$ and, for $j \ge 2$, the $j$th $\mathtt{b}$ is placed anywhere in $S[2\cdot 4^{j-2}+1\dd 4^{j-1}]$.
The family $\Sf_n^\mathtt{p}$ consists
of the length-$n$ prefixes of the infinite strings of the family $\Sf^p$, that is, $\Sf^\mathtt{p}_n = \{ S[1\dd n] : S \in \Sf^\mathtt{p} \}$.
\end{definition}

\begin{lemma}\label{lem:encode}
For every integer $n\ge 1$, the family $\Sf_n^\mathtt{p}$ needs $\Omega(\log^2 n)$ bits to be encoded.
\end{lemma}
\begin{proof}
In our definition of $\Sf^\mathtt{p}$, the location of the $j$th $\mathtt{b}$ can be chosen among $2\cdot 4^{j-2}$ positions, and each combination of
these choices generates a different string in $\Sf^\mathtt{p}_n$ as long as $n \ge 4^{j-1}$.
Hence, $|\Sf_n^\mathtt{p}| = \prod_{j=2}^{i+1} (2\cdot 4^{j-2}) = 2^{\Omega(i^2)}=2^{\Omega(\log^2 n)}$ for $i = \lfloor \log_4 n\rfloor$.
To distinguish strings in $\Sf_n^\mathtt{p}$, any encoding needs $\log |\Sf_n^\mathtt{p}| = \Omega(\log^2 n)$ bits.
\end{proof}

\begin{lemma}\label{lem:nsd}
For every length $n$ and integer $\delta \in [2\dd \lceil\frac{3n}{4}\rceil)$, there exists a family of length-$n$ strings of common measure~$\delta$
that needs $\Omega(\delta \log^2 \frac{n}{\delta})$ bits to be encoded.
\end{lemma}
\begin{proof}
As in the proof of \cref{lem:Sn}, we prove that the measure $\delta$ for any string $S \in \Sf_n^\mathtt{p}$ is at most 2. Starting from position $4^{j-1}+1$, the distance between any two consecutive $\mathtt{b}$s is at least $4^{j}$. Therefore, the 
distinct substrings of length $k\le 4^j$ are those that start at position $i\in [1\dd 4^{j-1}]$ and those of the form $\mathtt{a}^{k}$ or $\mathtt{a}^i \mathtt{b} \mathtt{a}^{k-i-1}$ for $i\in [0\dd k)$, which yields a total
of $d_{k}(S) \le 4^{j-1}+k+1$. Choosing $j=\lceil \log_4 k \rceil$, we get $d_k(S) \le 4^{\lceil \log_4 k \rceil-1}+k+1 \le 4^{\log_4 k}+k = 2k$. By definition of~$\delta$, we conclude that $\delta(S)\le 2$ for every $S\in \Sf^\mathtt{p}_n$. Thus, by Lemma~\ref{lem:encode}, encoding $\Sf^\mathtt{p}_n$ requires $\Omega(\delta\log^2 \frac{n}{\delta})$ bits.


We now generalize the result to larger $\delta$. As in the proof of \cref{thm:lbgamma}, let $m \ge 1$, $n \ge 4m-1$, and $n-m+1=\sum_{i=1}^m n_i$, where $n_i=\Omega(\frac{n}{m})$ and $n_i\ge 3$.
Let $S^{(i)}$, of length $n_i$, be built from some $S \in \Sf^\mathtt{p}_{n_i}$, with $\mathtt{a}$ replaced by $\mathtt{a}_i$ and $\mathtt{b}$ replaced by $\mathtt{b}_i$.
Finally, let $S^* = S^{(1)}\,\$_1\,S^{(2)}\,\$_2 \cdots \$_{m-1}\, S^{(m)}$. 
Since $d_k(S^{(i)}) \le 2k$ as per the previous paragraph, it holds, just as in the proof of \cref{thm:lbgamma}, that $\delta(S^*) = 3m-1$.
Let $\Sf^*_n$ be the set of possible strings $S^*$ of length $n$ we obtain by choosing the strings~$S^{(i)}$.
Even fixing the lengths $n_i$, we have $|\Sf^*_n|=\prod_{i=1}^m 2^{\Omega(\log^2 n_i)}$, and thus we need at least $\log |\Sf^*_n| = \sum_{i=1}^m \Omega(\log^2 n_i) = \Omega(m \log^2 \frac{n}{m}) = \Omega(\delta\log^2 \frac{n}{\delta})$ bits to encode a member of the family.
The case where $\delta < \frac34n$ is not of the form $3m-1$ is handled as in the proof of \cref{thm:lbgamma}.
\end{proof}


\begin{lemma}\label{lem:d}
For every length $n$ and integer $\delta \in [2\dd n]$, there exists a family of length-$n$ strings of common measure~$\delta$
that needs $\Omega(\delta \log \frac{n}{\delta}\log\delta)$ bits to be encoded.
\end{lemma}
\begin{proof}
Recall the string $S_n$ of \cref{def:S} and fix an integer $m\ge 1$.
Let $\Sf^\mathtt{r}_n\subseteq\{ \mathtt{a}, \mathtt{b}_1,\ldots,\allowbreak\mathtt{b}_m\}^n$ be the family of strings obtained from $S_n$ by replacing every $\mathtt{b}$ with $\mathtt{b}_r$ for some $r \in [1\dd m]$.
Since $|\Sf^\mathtt{r}_n| = m^{1+\lfloor{\log n}\rfloor}$, we need at least $\log |\Sf^\mathtt{r}_n| = \Omega(\log n\log m)$ bits to represent a string in $\Sf^\mathtt{r}_n$. 

On the other hand, as in the proof of \cref{lem:Sn}, the distinct substrings of length $k \le 2^j$ in $S^\mathtt{r}_n \in \Sf^\mathtt{r}_n$ are those starting in $S^\mathtt{r}_n[1\dd 2^{j-1}]$ and those of the form $\mathtt{a}^k$ or $\mathtt{a}^i\mathtt{b}_r\mathtt{a}^{k-i-1}$ for $i \in [0\dd k)$ and $r \in [1\dd m]$. Thus, $d_k(S^\mathtt{r}_n) \le 2^{j-1} + 1 + km$. Choosing $j=\lceil \log k \rceil$, we get $d_k(S^\mathtt{r}_n) \le k(m+1)$, and therefore $\delta(S^\mathtt{r}_n) \le m+1$.
 
Similarly to the proof of \cref{thm:lbgamma}, let $n \ge 4m-1$ and $n-3m+1=\sum_{i=1}^m n_i$, where $n_i=\Omega(\frac{n}{m})$ are positive integers. 
Let us choose any $S_{n_i}^{(i)} \in \Sf^\mathtt{r}_{n_i}$ and define
$S^* = S_{n_1}^{(1)}\,\$_1\,S_{n_2}^{(2)}\,\$_2 \cdots \$_{m-1}\, S_{n_m}^{(m)}\, S'$, where, in principle, $S'=\mathtt{a}^{2m}$. 
Consider the distinct length-$k$ substrings of $S^*$ for $k\le 2^j$. 
These include the (at most) $k(m-1)$ substrings containing a delimiter $\$_i$ (with $i\in [1\dd k)$)
and the (at most) $2^{j-1}\cdot m$ substrings starting within the first $2^{j-1}$ position of some $S_{n_i}^{(i)}$ (with $i\in [1\dd k]$). As argued in the previous paragraph, the remaining substrings are of the form
$\mathtt{a}^k$ or $\mathtt{a}^i\mathtt{b}_r\mathtt{a}^{k-i-1}$ for $r\in [1\dd m]$ and $i\in [1\dd k)$;
note that this analysis includes substrings overlapping $S_{n_m}^{(m)}$ and $S'=\mathtt{a}^{2m}$
despite the lack of delimiter between them.
Consequently, we get $d_k(S^*)\le k(m-1)+2^{j-1}m+km+1 \le (3m-1)k$, setting $j = \lceil \log k \rceil$.
At the same time, $m+1 \le d_1(S^*)\le 2m$, because $S^*$ contains $m-1$ delimiters, $\mathtt{a}$,
and between $1$ and $m$ distinct symbols $\mathtt{b}_r$ with $r\in [1\dd m]$.

We conclude that $m+1\le \delta(S^*)\le 3m-1$. We now modify $S'$ so that $\delta(S^*)=3m-1$.
For this, we replace the subsequent symbols $S'[2m], S'[2m-1], \ldots$ by fresh delimiters $\$_{m}, \$_{m+1}, \ldots$,
stopping as soon as $d_k(S^*)\ge (3m-1)k$ holds for some $k\in [1\dd n]$.
Since each value $d_k(S^*)$ grows by at most $1$ per delimiter added, we have $\delta(S^*)=3m-1$ upon termination of the process. Moreover, $d_1(S^*)$ grows by exactly $1$ per delimiter added, so the process terminates in at most $3m-1 - (m+1)=2m-2$ steps, which means that $S'$ is long enough to fit the delimiters.

We then obtain a family of possible strings $S^*[1\dd n]$ of common measure $\delta=3m-1$. Note that the suffix $S'$ is uniquely determined by the prefix containing the strings $S_{n_i}^{(i)}$.
Even fixing the lengths $n_i$, the number of possible strings $S^*$ we obtain by choosing the strings $S_{n_i}^{(i)}$ is $\prod_{i=1}^m |\Sf^\mathtt{r}_{n_i}| = \prod_{i=1}^m m^{1+\lfloor{\log n_i}\rfloor} > \prod_{i=1}^m m^{\log n_i}$. We thus need at least $\sum_{i=1}^m \log n_i\log m = 
\Omega(m \log \frac{n}{m}\log m) = \Omega(\delta\log \frac{n}{\delta}\log\delta)$ bits to encode a member of the family. 

The case where $\delta < \frac34n$ is not of the form $3m-1$ is handled as in the proof of \cref{thm:lbgamma}. 
For $\delta \ge \frac34n$, on the other hand, we construct a different family of length-$n$ strings with common measure $\delta$, with strings of the form $S=\$_{\pi(1)}\cdots\$_{\pi(\delta)} \,\cdot\,
\$_{\pi(\delta)}^{n-\delta}$ for a permutation $\pi$ of $[1\dd \delta]$. Each length $k\in [1\dd n]$ satisfies $d_k(S)=\min(\delta,n-k+1)$, so $\delta(S)=\delta$. Since there are $\delta!$ possible strings $S$, we need $\Omega(\delta\log\delta)=\Omega(\delta\log\frac{n}{\delta}\log\delta)$ bits to represent family members. 
\end{proof}

From the two lemmas above, we obtain the desired result:

\begin{theorem}\label{thm:both}
For every length $n$ and integer $\delta \in [2\dd n]$, there exists a family of length-$n$ strings of common measure $\delta$
that needs $\Omega(\delta \log \frac{n}{\delta}\log n)$ bits to be encoded.
\end{theorem}
\begin{proof}
Follows from Lemmas~\ref{lem:nsd} and~\ref{lem:d} since $\max(\log\delta,\log \frac{n}{\delta}) = \Omega(\log n)$
and since $\log \delta =\Omega(\log \frac{n}{\delta})$ for $\delta \ge \lceil\frac{3n}{4}\rceil$.
\end{proof}

\section{Bounds on grammar sizes}\label{sec:ub-gram}

In this section, we study the relation between $\delta$ and the sizes of the smallest context-free grammar and run-length context-free grammar generating $S$, denoted by $g$ and $g_{rl}$, respectively.
As observed in \cref{sec:delta}, \cref{lem:z} implies $g \le g_{rl} = O(\delta \log^2 \frac{n}{\delta})$.
Our first contribution is a lower bound construction for $g$:
for every length $n$ and value $\delta \in [2\dd n]$, we construct a string $S$
satisfying $g = \Omega(\delta \log^2\frac{n}{\delta} / \log \log \frac{n}{\delta})$.

Rather surprisingly, the situation for run-length context-free grammars is very different:
we prove that $g_{rl} = O(\delta \log \frac{n}{\delta})$, 
which is tight due to \cref{thm:lbgamma} and $g_{rl}=\Omega(\gamma)$.
Our argument is constructive and we derive a randomized algorithm
that, in $O(n)$ expected time, constructs a run-length context-free grammar of size $O(\delta\log\frac{n}{\delta})$ generating a given string $S$.

\subsection{A lower bound on grammar size}\label{sec:lb-gram}
\newcommand{\slp}{\mathsf{slp}}

For convenience, among all context-free grammars generating a single string $S$,
we only consider \emph{straight-line programs} (SLPs), where the right-hand side
of each production is of size exactly 2.
We denote by $\slp(S)$ the minimum number of symbols (terminals and non-terminals)
in any SLP generating $S$.
Each context-free grammar generating $S$ can be transformed to an SLP,
and thus $\slp(S) = \Theta(g(S))$ holds for every string $S$.

Our construction relies on the family $\Sf^{\mathtt{p}}_n$ of \cref{def:Sstar} and the following consequence of \cref{lem:encode}.
\begin{corollary}\label{cor:slp}
For every integer $n\ge 1$, there exists a string $S^{\mathtt{p}}_n \in \Sf^{\mathtt{p}}_n$
satisfying $\slp(S^{\mathtt{p}}_n) = \Omega(\log^2 n / \log \log n)$.
\end{corollary}
\begin{proof}
    Recall that each string $S\in \Sf^{\mathtt{p}}_n$ is binary and therefore
    can be encoded in $O(\slp(S) \log \slp(S))$ bits: every symbol is assigned
    a $\lceil \log \slp(S) \rceil$-bit identifier (with $0$ and $1$ reserved for $\mathtt{a}$ and $\mathtt{b}$, respectively) so that each production is encoded using $O(\log \slp(S))$ bits. 
    At the same time, \cref{lem:encode} proves that every encoding distinguishing members of  $\Sf^{\mathtt{p}}_n$
    requires $\Omega(\log^2 n)$ bits. In particular, our SLP-based encoding uses $\Omega(\log^2 n)$ bits
    for some string $S^{\mathtt{p}}_n \in \Sf^{\mathtt{p}}_n$.
    This string satisfies  $\slp(S^{\mathtt{p}}_n) \log \slp(S^{\mathtt{p}}_n) = \Omega(\log^2 n)$, and therefore
    $\slp(S^{\mathtt{p}}_n) = \Omega(\log^2 n / \log \log n)$.
\end{proof}
Note that the proof of \cref{cor:slp} does not apply to run-length context-free grammars because 
a production of the form $A\rightarrow B^t$ may need $\Theta(\log t)$ bits to represent the exponent $t$.
In fact, since $\delta(S^{\mathtt{p}}_n)=2$ holds for $n\ge 3$, the upper bound $g_{rl}=O(\delta \log \frac{n}{\delta})$ proved in the next section shows that $g_{rl}(S^{\mathtt{p}}_n)=O(\log n)$; generalizing \cref{cor:slp} to run-length context-free grammars is therefore inherently impossible.

\cref{cor:slp} shows that $g =\Omega(\delta \log^2\frac{n}{\delta} / \log \log \frac{n}{\delta})$
is possible for $\delta=2$ and any length $n\ge 3$. We generalize this construction to arbitrary $\delta\in [2\dd n]$ as in the proof of \cref{lem:nsd}. Our argument requires the following auxiliary lemma.

\begin{lemma}\label{lem:slps}
Consider a string $S = L \cdot R$. If the alphabets $\Sigma(L)$ of $L$ and $\Sigma(R)$ of $R$ are disjoint,
then $\slp(S) \ge \slp(L) + \slp(R)$.
\end{lemma}
\begin{proof}
Our goal is to construct SLPs generating $L$ and $R$ with $\slp(S)$ symbols in total.
Let us fix an SLP generating $S$ with $\slp(S)$ symbols.
To generate $L$ ($R$), we start from the SLP of $S$ and remove every terminal not in $\Sigma(L)$ ($\Sigma(R)$).
Then, transitively remove nonterminals with empty right-hand sides and suppress nonterminals with right-hand sides of length 1. 

Now consider a rule $A \rightarrow BC$ in the SLP of $S$: either $C$ disappears in the SLP that generates $L$, or $B$ disappears in the SLP that generates $R$, or both. In all cases, the rule survives in at most one of the two SLPs; therefore the total number of symbols in the SLPs we build for $L$ and $R$ is at most $\slp(S)$. 
Consequently, $\slp(L) + \slp(R) \le \slp(S)$.
\end{proof}

\begin{theorem}\label{thm:gvsdelta}
For every length $n$ and integer value $\delta\in [2\dd n]$, there is a string $S[1\dd n]$ with
$g=\Omega(\delta\log^2\frac{n}{\delta}/\log\log\frac{n}{\delta})$. 
\end{theorem}
\begin{proof}
As in the proof of \cref{lem:nsd}, let $m \ge 1$, $n \ge 4m-1$, and $n-m+1=\sum_{i=1}^m n_i$, where $n_i=\Omega(\frac{n}{m})$ and $n_i\ge 3$.
Moreover, let $S = S^{(1)}\,\$_1\,S^{(2)}\,\$_2 \cdots \$_{m-1}\, S^{(m)}$,
where $S^{(i)}$ is obtained from the string $S^\mathtt{p}_{n_i}\in \Sf^\mathtt{p}_{n_i}$ of \cref{cor:slp} by replacing every $\mathtt{a}$ with $\mathtt{a}_i$ and every $\mathtt{b}$ with $\mathtt{b}_i$.
Note that $S$ belongs to the family constructed in the proof of \cref{lem:nsd}, and thus $\delta(S)=3m-1$.
Furthermore, since the alphabets of strings $S^{(i)}$ (for $i\in [1\dd m]$) and $\$_i$ (for $i\in [1\dd m)$)
are pairwise disjoint, \cref{lem:slps} yields $\slp(S)\ge \sum_{i=1}^m \slp(S^{(i)}) = \sum_{i=1}^m \slp(S^\mathtt{p}_{n_i}) = \Omega(m\log^2\frac{n}{m}/\log\log\frac{n}{m})$. 

This construction proves the theorem for $\delta = 3m-1$ and $n\ge 4m-1$.
The case where $\delta < \frac34n$ is not of the form $3m-1$ is handled as in the proof of \cref{thm:lbgamma}.
Finally, we note that if $\delta \ge \frac34n = \Omega(n)$, then the claim reduces to $g = \Omega(\delta)$ and therefore follows from \cref{lem:delta} due to $\gamma = O(g)$.
\end{proof}

\subsection{An upper bound on run-length grammar size}\label{sec:ub}
In this section, we prove that every string $S\in \Sigma^n$ can be generated using
a run-length context-free grammar of size $\Oh(\delta \log \frac{n}{\delta})$.
We obtain our grammar in the process of building a locally consistent parsing on top of $S$.
Our parsing is based on the recompression technique by Jeż~\cite{Jez16},
who used it to design a simple $\Oh(n)$-time algorithm constructing a (standard)
context-free grammar of size $\Oh(z \log \frac{n}{z})$.
More specifically, we rely on a very recent \emph{restricted recompression} by Kociumaka et al.~\cite{IPM},
which delays processing symbols generating long substrings.
Birenzwige et al.~\cite{Birenzwige2020} applied a similar idea to transform the locally consistent parsing
of Sahinalp and Vishkin~\cite{FirstConsistentParsing}.

\subsubsection{Run-length grammar construction via restricted recompression}
Both recompression and restricted recompression, given a string $S\in \Sigma^+$,
construct a sequence of strings $(S_k)_{k=0}^\infty$ over the alphabet $\Symb$ of \emph{symbols}
 defined as the least fixed point of the following equation:
 \[\Symb = \Sigma \cup (\Symb \times \Symb)\cup (\Symb \times \mathbb{Z}_{\ge 2} ).\]
Symbols in $\Symb \sm \Sigma$ are \emph{non-terminals}
with \emph{productions} $(A_1,A_2)\to A_1A_2$ for $(A_1,A_2)\in \Symb \times \Symb$ and $(A_1,m) \to A_1^m$
for $(A_1,m)\in \Symb\times \mathbb{Z}_{\ge 2}$.
With any $A\in \Symb$ designated as the start symbol, this yields a run-length straight-line program (RLSLP).
The following \emph{expansion} function $\val: \Symb \to \Sigma^+$ retrieves the string generated by this RLSLP:
\[\val(A) = \begin{cases}
  A & \text{if $A\in \Sigma$},\\
  \val(A_1)\cdot \val(A_2) & \text{if $A=(A_1,A_2)$ for $A_1,A_2\in \Symb$},\\
  \val(A_1)^m & \text{if $A=(A_1,m)$ for $A_1\in \Symb$ and $m\in \mathbb{Z}_{\ge 2}$}.
\end{cases}\]
Intuitively, $\Symb$ forms a \emph{universal} RLSLP:
for every RLSLP with symbols $\Pres$ and non-terminals $\Sigma\sub \Pres$, there is a unique homomorphism $f:\Pres \to \Symb$ such that $f(A)=A$ if $A\in \Sigma$, $f(A)\to f(A_1)f(A_2)$ if $A\to A_1A_2$, and $f(A)\to f(A_1)^m$ if $A\to A_1^m$.
As a result, $\Symb$ provides a convenient formalism to argue about procedures generating RLSLPs.

The main property of strings $(S_k)_{k=0}^\infty$ generated using (restricted) recompression
is that $\val(S_k)=S$ holds for all $k\in \Zz$, where $\val: \Symb\to \Sigma^+$ is lifted to $\val : \Symb^*\to \Sigma^*$ by setting $\val(A_1\cdots A_a)=\val(A_1)\cdots \val(A_a)$ for $A_1\cdots A_a\in \Symb^*$.
The subsequent strings $S_k$, starting from $S_0=S$, are obtained by alternate applications 
of the following two functions which
decompose a string $T\in \Symb^+$ into \emph{blocks} and then \emph{collapse} blocks
into appropriate symbols. In \cref{it:rle}, all blocks of length at least $2$ are maximal blocks of the same symbol,
and they are collapsed to symbols in $\Symb\times \mathbb{Z}_{\ge 2}$.
In \cref{it:pair}, there are no blocks of length more than 2, and all blocks of length 2 are collapsed to symbols in $\Symb \times \Symb$.

\begin{definition}[Restricted run-length encoding~\cite{IPM}]\label{it:rle}
Given $T\in \Symb^+$ and $\Act \sub \Symb$, we define $\rle_{\Act}(T)\in \Symb^+$
to be the string obtained as follows by decomposing $T$ into blocks and collapsing these blocks:
\begin{enumerate}
  \item For $i\in [1\dd |T|)$, place a \emph{block boundary} between $T[i]$ and $T[i+1]$
  unless $T[i]=T[i+1]\in \Act$.
  \item Replace each block $T[i\dd i+m)= A^m$ of length $m\ge 2$ with a symbol $(A,m)\in \Symb$.
\end{enumerate}
\end{definition}

\begin{definition}[Restricted pair compression~\cite{IPM}]\label{it:pair}
  Given $T\in \Symb^+$ and disjoint sets $\Left,\Right \sub \Symb$, we define $\pc_{\Left,\Right}(T)\in \Symb^+$
  to be the string obtained as follows by decomposing $T$ into blocks and collapsing these blocks:
  \begin{enumerate}
    \item For $i\in [1\dd |T|)$, place a \emph{block boundary} between $T[i]$ and $T[i+1]$
    unless $T[i]\in \Left$ and $T[i+1]\in \Right$.
    \item Replace each block $T[i\dd i+1]$ of length $2$ with a symbol $(T[i],T[i+1])\in \Symb$.
  \end{enumerate}
\end{definition}
The original recompression uses $\rle_{\Act}$ with $\Act = \Symb$ and $\pc_{\Left,\Right}$ with $\Symb = \Left\cup \Right$. 
In the restricted version, symbols in $\Symb\sm \Act$ and $\Symb\sm (\Left \cup \Right)$, respectively,
are forced to form length-$1$ blocks.
In the $k$th round of restricted recompression, this mechanism is applied to symbols $A$ whose expansion $\val(A)$ is longer than a certain threshold $\ell_k$. 

With this intuition, we are now ready to formally define the sequence $(S_k)_{k=0}^\infty$
constructed through restricted recompression.

\begin{construction}[Restricted recompression~\cite{IPM}]\label{constr:Jez}
Given a string $S\in \Sigma^+$, the strings $S_k$ for $k\in \Zz$ are constructed as follows,
based on $\ell_k := (\tfrac{8}{7})^{\ceil{k/2}-1}$ and $\Symb_k := \{A\in \Symb : |\val(A)|\le \ell_k\}$:
\begin{itemize}
  \item If $k=0$, then $S_k = S$.
  \item If $k>0$ is odd, then $S_{k}=\rle_{\Symb_k}(S_{k-1})$.
  \item If $k>0$ is even, then $S_{k}=\pc_{\Left_k,\Right_k}(S_{k-1})$, where $\Symb_k = \Left_k\cup \Right_k$ is a uniformly random partition into disjoint classes.
\end{itemize}
\end{construction}

It is easy to see that $\val(S_k)=S$ indeed holds for all $k\in \Zz$.
As we argue below, almost surely (with probability~1),
there exists $h\in \Zz$ such that $|S_h|=1$. 
In particular, an RLSLP generating $S$ can be obtained by setting $S_h[1]$ as the starting symbol of the RLSLP derived from by $\Symb$.
While this RLSLP contains infinitely many symbols, it turns out that we can remove symbols that do not occur in any string $S_k$.
Formally, for each $k\in \Zz$, let us define the family $\Pres_k:=\{S_k[j] : j\in [1\dd |S_k|]\}\sub \Symb$
of symbols occurring in $S_k$. 
Observe that each symbol in $S_0$ belongs to $\Sigma$ and, for $k\in \Zp$, each symbol in $S_k$
was either copied from $S_{k-1}$ or obtained by collapsing a block in $S_{k-1}$.
Consequently, the family $\Pres := \bigcup_{k=0}^\infty \Pres_k$ satisfies $\Pres \sub \Sigma \cup (\Pres \times \Pres)\cup (\Pres\times \mathbb{Z}_{\ge 2})$.
In other words, for every non-terminal $A\in \Pres$, the symbols on the right-hand side of the production of $A$ also belong to $\Pres$. Hence, the remaining symbols can indeed be removed from the RLSLP generating $S$.
As  the size of resulting run-length context-free grammar is proportional to $|\Pres|$,
we are left with the task of bounding $\Exp[|\Pres|]$.

\subsubsection{Analysis of the grammar size}
Our argument relies on several properties of restricted recompression proved in~\cite{IPM}.
Due to the current status of~\cite{IPM} being an unpublished manuscript, the proofs are provided in \cref{app:proofs}
for completeness.

\begin{restatable}[\cite{IPM}]{fact}{fctcons}\label{fct:cons}
  For every $k\in \Zz$, if $\val(x)=\val(x')$ holds for two fragments of $S_k$,
  then $x= x'$.
 \end{restatable}
  
\begin{restatable}[\cite{IPM}]{corollary}{cordistinct}\label{cor:distinct}
  For every odd $k\in \Zz$, there is no $j\in [1\dd |S_k|)$ such that $S_k[j]=S_k[j+1]\in \Symb_{k+1}$.
\end{restatable}

Recall that $\val(S_k)=S$ for every $k\in \Zz$. Hence, for every $j\in [1\dd |S_k|]$,
we can associate $S_k[j]$ with a fragment $S(|\val(S_k[1\dd j))|\dd |\val(S_k[1\dd j])|]= \val(S_k[j])$;
these fragments are called \emph{phrases} (of $S$) induced by $S_k$.
We also define a set $B_k$ of \emph{phrase boundaries} induced by $S_k$:
\[B_k = \{|\val(S_{k}[1\dd j])| : j\in [1\dd |S_k|)\}.\]

\begin{restatable}[\cite{IPM}]{lemma}{lemrecompr}\label{lem:recompr1}
Let $\alpha \in \mathbb{Z}_{\ge 1}$ and let $i,i'\in [\alpha \dd n-\alpha]$ be such that $S(i-\alpha\dd i+\alpha]= S(i'-\alpha\dd i'+\alpha]$.
For every $k\in \Zz$, if $\alpha \ge 16\ell_k$, then $i\in B_k \Longleftrightarrow i'\in B_k$.
\end{restatable}
\begin{restatable}[\cite{IPM}]{lemma}{lemrand}\label{lem:rand0}
  For every $k\in \Zz$, we have $\Exp[|S_k|]< 1+\frac{4n}{\ell_{k+1}}$.
\end{restatable}

\cref{lem:rand0} can be used to confirm that  almost surely $|S_k|=1$ holds for some $k\in \Zz$.
\begin{corollary}\label{cor:fin}
With probability $1$, there exists $k\in \Zz$ such that $|S_k|=1$.
\end{corollary}
\begin{proof}
For a proof by contradiction, suppose that $\eps := \Pr[\min_{k=0}^\infty |S_k| > 1] > 0$.
In particular, this yields $\Exp[|S_k|]\ge 1+\eps$ for every $k\in \Zz$.
However, \cref{lem:rand0} implies $\lim_{k\to \infty} |S_k| \le 1 + \lim_{k\to \infty} \frac{4n}{\ell_{k+1}}=1$,
a contradiction.
\end{proof}

Our main goal is to prove that $\Exp[|\Pres|]=\Oh(\delta \log \frac{n}{\delta})$ (\cref{cor:all}).
As a stepping stone, we show that $\Exp[|\Symb_{k+1}\cap \Pres_k|] = \Oh(\delta)$ holds for all $k\in \Zz$ (\cref{lem:level}).
The restriction to symbols in $\Symb_{k+1}$ is not harmful because every symbol in $S_k$ that does not belong to $\Symb_{k+1}$ forms a length-$1$ block that gets propagated to $S_{k+1}$.
The strategy behind the proof of \cref{lem:level} is to consider the phrases induced by the leftmost occurrences of all symbol in $\Symb_{k+1} \cap \Pres_k$. 
Using \cref{lem:recompr1}, we construct $\Oh(\delta)$ fragments of $S$ of total length $\Oh(\ell_k\delta)$ guaranteed to overlap all these phrases, and we show that these fragments in expectation overlap $\Oh(\delta)$
phrases induced by $S_k$.
The latter claim is a consequence of the following generalization of \cref{lem:rand0}.
\begin{lemma}\label{lem:rand}
For every $k\in \Zz$ and every interval $I\sub [1\dd n)$,
we have \[\Exp[|B_k\cap I|]< 1+\tfrac{4|I|}{\ell_{k+1}}.\]
\end{lemma}
\begin{proof}
  We proceed by induction on $k$. For $k=0$, we have $|B_k\cap I|=|I| < 1+4|I| = 1+\frac{4|I|}{\ell_1}$.
  If $k$ is odd, we note that $B_{k}\sub B_{k-1}$ and therefore $\Exp[|B_k\cap I|] \le \Exp[|B_{k-1}\cap I|]
  < 1+\frac{4|I|}{\ell_{k}} = 1+\frac{4|I|}{\ell_{k+1}}$.
  Thus, it remains to consider even values $k>0$.

  \begin{claim}\label{clm:rand}
    If $k>0$ is even, then, conditioned on any fixed $S_{k-1}$, we have $
      \Exp\left[|B_{k}\cap I| \bigm| S_{k-1}\right]< \tfrac14 + \tfrac{|I|}{2\ell_k}+\tfrac34|B_{k-1}\cap I|$.
  \end{claim}
  \begin{proof}
    Let us define 
    \begin{align*} 
      J &= \{j\in [1\dd |S_{k-1}|) :S_{k-1}[j]\notin\Symb_k\text{ or }S_{k-1}[j+1]\notin\Symb_k\},\\
      J_I &= \{j \in J : |\val(S_{k-1}[1\dd j])|\in I\}\sub B_{k-1}\cap I.
    \end{align*}
  Since $A\notin \Symb_k$ yields $|\val(A)|>\ell_k$, we have $|J_I| < 1+\tfrac{2|I|}{\ell_k}$.
  Moreover, observe that if $j\in [1\dd |S_{k-1}|)\sm J$, then $S_{k-1}[j]$ and $S_{k-1}[j+1]$ are,
  by \cref{cor:distinct}, distinct symbols in $\Symb_k$. Consequently,
  \[\Pr[S_{k-1}[j]\in \Left_k \text{ and } S_{k-1}[j+1]\in \Right_k] = \tfrac14.\]
  Thus, the probability that $\pc_{\Left_k,\Right_k}(S_{k-1})$ places a block boundary after position $j\in [1\dd |S_{k-1}|)\sm J$
  is $\frac34$. Therefore,
  \[\Exp\left[|B_{k}\cap I| \bigm| S_{k-1}\right] =  |J_I| + \tfrac34(|B_{k-1}\cap I|-|J_I|)
     = \tfrac14|J_I|+\tfrac34|B_{k-1}\cap I|  <\tfrac14 + \tfrac{|I|}{2\ell_k}+\tfrac34|B_{k-1}\cap I|.\qedhere
    \]
  \end{proof}

  Since the partition $\Symb_k=\Left_k\cup\Right_k$ is independent of $S_{k-1}$,
  \cref{clm:rand} and the inductive assumption yield
  \[\Exp[|B_k\cap I|] < \tfrac14 + \tfrac{|I|}{2\ell_k}+\tfrac34\Exp[|B_{k-1}\cap I|] < \tfrac14 + \tfrac{|I|}{2\ell_k} + \tfrac34 + \tfrac{3|I|}{\ell_k}=1+\tfrac{7|I|}{2\ell_k}=1+\tfrac{4|I|}{\ell_{k+1}}.\qedhere\]
\end{proof}

Next, we apply \cref{lem:rand,lem:recompr1} to bound the expected 
size of $\Pres_k\cap \Symb_{k+1}$.

\begin{lemma}\label{lem:level}
  For every $k\in \Zz$ and every string $S\in \Sigma^+$ with measure $\delta$, we have $\Exp[|\Pres_k\cap \Symb_{k+1}|]=\Oh(\delta)$.
\end{lemma}
\begin{proof}
   Let us fix integers $\alpha \ge 16\ell_k$ and $m = 2\alpha + \floor{\ell_{k+1}}$.
  Moreover, define 
  \[L = \{i\in [0\dd n-m] : S(i\dd i+m] = S(i'\dd i'+m]\text{ for some }i'\in[0\dd i)\}\]
  and 
  \[P = \{i\in [1\dd n] : i-\alpha\notin L\}.\]

  \begin{claim}\label{clm:ub}
    We have $|\Pres_k\cap \Symb_{k+1}| \le 1+|B_k\cap P|$.
  \end{claim}
  \begin{proof}
    Let $S_k[j]$ be the leftmost occurrence in $S_k$ of $A\in \Symb_{k+1}\cap \Pres_k$.
    Moreover, let $p=|\val(S_k[1\dd j))|$ and $q=|\val(S_k[1\dd j])|$
    so that $S(p\dd q]=\val(A)$ is the phrase induced by $S_k$ corresponding to $S_k[j]$.

    We shall prove that $j=1$ or $p\in B_k\cap P$.
    This will complete the proof of the claim because 
    distinct 
    symbols $A$ yield distinct positions $j$ and $p$.

    For a proof by contradiction, suppose that $j\in (1\dd |S_k|)$
    yet $p\notin B_k\cap P$.
    Since $p\in B_k$ holds due to $j>1$, we derive $p\notin P$, which
    implies $p-\alpha \in L$.
    Consequently, there is a position $p'\in [\alpha\dd p)$
    such that $S(p-\alpha \dd p-\alpha+m]=S(p'-\alpha\dd p'-\alpha+m]$.
    In particular, $S(p-\alpha \dd p+\alpha]=S(p'-\alpha\dd p'+\alpha]$,
    so \cref{lem:recompr1} yields $p'\in B_k$.
    Similarly, due to $q-p=|\val(A)|\le \floor{\ell_{k+1}}=m-2\alpha$,
    we have $S(q-\alpha \dd q+\alpha]=S(q'-\alpha\dd q'+\alpha]$
    for $q':=p'+|\val(A)|$, and therefore $q'\in B_k$ holds due to $q\in B_k$.
    \cref{lem:recompr1} further implies $B_k\cap(p'\dd q') = \emptyset = B_k\cap (p\dd q)$.
    Consequently, $S(p'\dd q']$ is a phrase induced by $S_k$,
    and, since $p'<p$, it corresponds to $S_k[j']$ for some $j'<j$.
    By \cref{fct:cons}, we have $S_k[j']=S_k[j]=A$, which contradicts the choice of $S_k[j]$
    as the leftmost occurrence of $A$ in $S_k$.
  \end{proof}

  Consequently, it remains to prove that $\Exp[|B_k\cap P|]  = \Oh(\delta)$.
  For this, we characterize $P$ as follows.

  \begin{claim}\label{clm:cov}
    The set $P$ can be covered by $\Oh(\delta)$ intervals of total length $\Oh(m\delta)$.
  \end{claim}
  \begin{proof}
    Note that $P' = \bigcup_{i\in P} (i-\alpha\dd i+m-\alpha]$ is a superset of $P$. 
    Moreover, each position $j\in P'\cap [m\dd n-m]$ is contained in the leftmost occurrence of a length-$m$
    substring of $S$ and then  $S(j-m\dd j+m]$ is the leftmost occurrence of a length-$2m$ substring of $S$.
    Consequently, $|P'\cap [m\dd n-m]| \le 2m\delta$.
    Since $P'\setminus [m\dd n-m] = (1-\alpha \dd m)\cup (n-m\dd n+m-\alpha]$ is of size $\Oh(m)$,
    we conclude that $|P'|=\Oh(m\delta)$.
    Next, recall that $P'$ is a union of length-$m$ integer intervals.
    Merging overlapping intervals, we get a decomposition into disjoint intervals
    of length at least $m$. The number of intervals does not exceed $\frac1m |P'| = \Oh(\delta)$.
  \end{proof}

  Now, let $\mathcal{I}$ be the family of intervals covering $P$ obtained using \cref{clm:cov}. 
  For each $I\in \mathcal{I}$, \cref{lem:rand} implies $\Exp[|B_k\cap I|]\le 1+\frac{4|I|}{\ell_{k+1}}$.
  By linearity of expectation and the bounds in \cref{clm:cov}, this yields the claimed result:
  \[\Exp[|B_k\cap P|] \le |\mathcal{I}|+\tfrac{4}{\ell_{k+1}}\sum_{I\in \mathcal{I}}|I|=\Oh(\delta + \tfrac{4\delta m }{\ell_{k+1}})=\Oh(\delta).\qedhere\]
\end{proof}

The proof of our main bound $\Exp[|\Pres|]=\Oh(\delta\log \frac{n}{\delta})$ combines \cref{lem:level,lem:rand0}.
\begin{corollary}\label{cor:all}
  For every string $S$ of length $n$ and measure $\delta$, we have
  $\Exp[|\Pres|]=\Oh(\delta \log\frac{n}{\delta})$.
\end{corollary}
\begin{proof}
  Note that $|\Pres|\le 1 + \sum_{k=0}^\infty |\Pres_{k}\sm \Pres_{k+1}|$.
  We combine two upper bounds on $\Exp[|\Pres_{k}\sm \Pres_{k+1}|]$, following from \cref{lem:level,lem:rand0}, respectively. 

  First, we observe that \cref{constr:Jez} guarantees $\Pres_{k}\sm \Pres_{k+1}\sub \Pres_k \cap \Symb_{k+1}$ and thus $\Exp[|\Pres_{k}\sm \Pres_{k+1}|]\le \Exp[|\Pres_{k}\cap \Symb_{k+1}|]=\Oh(\delta)$ holds due to \cref{lem:level}.
  Secondly, we note that $|\Pres_k\sm \Pres_{k+1}| = 0$ if $|S_k|=1$
  and $|\Pres_k\sm \Pres_{k+1}| \le |S_k|$ otherwise.
  Consequently, Markov inequality and \cref{lem:rand0} yield
  \begin{multline*}
    \Exp[|\Pres_{k}\sm \Pres_{k+1}|] \le \Exp[|S_k|] - \mathbb{P}[|S_k|=1]
    = \Exp[|S_k|]-1 + \mathbb{P}[|S_k|\ge 2] 
    = \Exp[|S_k|]-1 + \mathbb{P}[|S_k|-1 \ge 1]\\
    \le \Exp[|S_k|] -1 +  \Exp[|S_k|-1]
     = 2\Exp[|S_k|]-2\le \tfrac{8n}{\ell_{k+1}}.
  \end{multline*}
  Thus, $ \Exp[|\Pres_{k}\sm \Pres_{k+1}|] = \Oh((\frac78)^{k/2}n)$.

  We apply the first or the second upper bound on $\Exp[|\Pres_{k}\sm \Pres_{k+1}|]$ depending on whether $k \ge 2\log_{8/7}\frac{n}{\delta}$.
  This yields\[
    \sum_{k = 0}^\infty \Exp[|\Pres_{k}\sm \Pres_{k+1}|] = \log_{8/7}\tfrac{n}{\delta}\cdot \Oh(\delta) + \sum_{i=0}^\infty \Oh\left((\tfrac78)^{i/2}\delta\right) = \Oh(\delta \log\tfrac{n}{\delta}).\]
  Consequently, $\Exp[|\Pres|] = 1 + \Oh(\delta \log\tfrac{n}{\delta}) = \Oh(\delta \log\tfrac{n}{\delta})$ holds as claimed.
\end{proof}

Finally, we note that \cref{cor:fin,cor:all} allow bounding the size of the smallest run-length grammar generating~$S$.

\begin{theorem}\label{thm:rlslp}
Every string $S[1\dd n]$ satisfies $g_{rl}=\Oh(\delta \log\frac{n}{\delta})$.
\end{theorem}
\begin{proof}
  We apply \cref{constr:Jez} on top of the given string $S$.
  By \cref{cor:fin,cor:all}, the random choices within \cref{constr:Jez} can be fixed so that $|\Pres| = \Oh(\delta \log \frac{n}{\delta})$ and $|S_k|=1$ holds for sufficiently large $k$.
  We build a run-length grammar with symbols $\Pres$.
  Each symbol $A\in \Sigma$ is a terminal symbol, each symbol $A=(A_1,A_2) \in \Symb\times \Symb$
  is associated with a production $A \to A_1A_2$,
  and each symbol $A=(A_1,m)\in \Symb\times \mathbb{Z}_{\ge 2}$ is associated with a production $A \to A_1^m$.
  It easy to see that the auxiliary symbols $A_1,A_2$ belong to $\Pres$
  and that the expansion of each symbol $A$ (within the grammar) is $\exp(A)$.
  In particular, if we set the only symbol of $S_k$ for sufficiently large $k$ as the starting symbol,
  then the grammar generates~$S$.
\end{proof}

\subsection{Efficient construction of a small run-length grammar}\label{sec:alg}
In this section, we convert the proof of \cref{thm:rlslp} to a fast algorithm generating
a run-length grammar of size $\Oh(\delta\log \frac{n}{\delta})$.

\begin{proposition}\label{prp:rlslp}
  There exists a randomized algorithm that, given a string $S$ of length $n$ and measure $\delta$,
  in $\Oh(n)$ expected time constructs a run-length grammar of expected size $\Oh(\delta \log \frac{n}{\delta})$ generating $S$.
\end{proposition}
\begin{proof}
The algorithm simulates \cref{constr:Jez} constructing $\Pres$, which can be interpreted as a grammar generating $S$ (see the proof of \cref{thm:rlslp}).
Each symbol $A\in \Sigma$ is stored explicitly, each symbol $A=(A_1,A_2)\in \Symb\times \Symb$
keeps pointers to $A_1$ and $A_2$, and each symbol  $A=(A_1,m)\in \Symb\times \mathbb{Z}_{\ge 2}$
keeps $m$ and a pointer to $A_1$. Additionally, each symbol $A$ is augmented with $|\val(A)|$
and an \emph{identifier} $\id_k(A)$ for every $k$ such that $A\in \Pres_k$,
where $\id_k : \Pres_k \to [1\dd |\Pres_k|]$ is a bijection.
The inverse mappings $\id_k^{-1}$ are implemented as arrays.

The string $S_0 = S$ is given as input. In order to construct $\Pres_0$ and $\id_0$,
we sort the characters of $S$ and assign them consecutive positive integer identifiers.
This step takes $\Oh(n)$ time due to the assumption $\sigma = n^{\Oh(1)}$. 

In order to construct $S_k$, $\Pres_k$, and $\id_k$, we process $S_{k-1}$ depending on the parity of $k$.
If $k$ is odd, we scan $S_{k-1}$ from left to right outputting subsequent symbols of $S_{k}$.
Initially, each symbol $A\in \Pres_k$ is represented as $(\id_{k-1}(A_1),m)$ (if $A=(A_1,m)\in \Pres_{k}\sm \Pres_{k-1}$) or as $\id_{k-1}(A)$ (otherwise).
Suppose that $S_{k-1}[j\dd ]$ is yet to be processed.
If $|\val(S_{k-1}[j])| > \ell_{k}$ or $S_{k-1}[j]\ne S_{k-1}[j+1]$,
we output $\id_{k-1}(S_{k-1}[j])$ as the next symbol of $S_{k}$ and continue processing $S_{k-1}[j+1\dd ]$.
Otherwise, we determine the maximum integer $m\ge 2$ such that $S_{k-1}[j']=S_{k-1}[j]$ for $j'\in [j\dd j+m)$,
output $(\id_{k-1}(S_{k-1}[j]),m)$ as the next symbol of~$S_{k}$, and continue processing $S_{k-1}[j+m\dd ]$.
(By \cref{fct:cons}, $(S_{k-1}[j],m) \notin \Pres_{k-1}$ in this case.)
Note that the symbols in $\Pres_k$ are initially represented as elements of $[1\dd |S_{k-1}|]\cup [1\dd |S_{k-1}|]^2$. Sorting these values allows 
constructing symbols in $\Pres_k\sm \Pres_{k-1}$ and the identifier function $\id_k$ in $\Oh(|S_{k-1}|)$ time.

If $k$ is even, we first randomly partition $\Symb_k$ into $\Left_k$ and $\Right_k$.
Technically, this step consists in iterating over symbols $A\in \Pres_{k-1}$
and appropriately marking $A$ if $|\val(A)|\le \ell_k$.
Next, we scan $S_{k-1}$ from left to right outputting subsequent symbols of $S_{k}$.
Initially, each symbol $A$ in $S_{k}$ is represented as $(\id_{k-1}(A_1),\id_{k-1}(A_2))$ (if $A=(A_1,A_2)\notin \Pres_{k-1}$)
or as $\id_{k-1}(A)$ (otherwise).
Suppose that $S_{k-1}[j\dd ]$ is yet to be processed.
If $S_{k-1}[j]\notin \Left_k$ or $S_{k-1}[j+1]\notin \Right_k$,
we output $\id_{k-1}(S_{k-1}[j])$ as the next symbol of $S_{k}$ and continue processing $S_{k-1}[j+1\dd ]$.
Otherwise, we output $(\id_{k-1}(S_{k-1}[j]),\id_{k-1}(S_{k-1}[j+1]))$ as the next symbol of $S_{k}$ and continue processing $S_{k-1}[j+2\dd ]$.
(By \cref{fct:cons}, $(S_{k-1}[j],S_{k-1}[j+1])\notin \Pres_{k-1}$ in this case.)
Note that $|S_{k}|\le |S_{k-1}|$ and that the characters of $S_{k}$ are initially represented as elements of $[1\dd |S_{k-1}|]\cup [1\dd |S_{k-1}|]^2$. 
Sorting these values allows 
constructing symbols in $\Pres_k\sm \Pres_{k-1}$ and the identifier function $\id_k$ in $\Oh(|S_{k-1}|)$ time.

The algorithm terminates when it encounters a string $S_h$ with $|S_h|=1$, marking the only symbol of $S_h$ as the starting symbol of the grammar.
The overall running time is $\Oh(\sum_{k=0}^h |S_k|)$, which is $\Oh(\sum_{k=0}^h \frac{n}{\ell_{k+1}})=\Oh(n)$ in expectation by \cref{lem:rand}. The expected grammar size is $\Oh(\delta \log \frac{n}{\delta})$ by \cref{cor:all}.
\end{proof}

Finally, we adapt the algorithm to output a small grammar in the worst case.
\begin{theorem}\label{thm:grl}
  There exists a randomized algorithm that, given a string $S$ of length $n$ and measure $\delta$,
  constructs a run-length grammar of size $\Oh(\delta \log \frac{n}{\delta})$ generating $S$.
  The running time is $\Oh(n)$ in expectation and $\Oh(n\log n)$ w.h.p.
\end{theorem}
\begin{proof}
  We compute $\delta$ using \cref{lem:compute delta}
  and determine an upper bound on the expected size of the grammar produced using \cref{prp:rlslp},
  as well as an upper bound on the expected running time of the algorithm of \cref{prp:rlslp}.

  Then, we repeatedly call the algorithm of \cref{prp:rlslp} with a time limit of 4 times the expected running time.
  If the call does not finish within the limit, we interrupt the execution and proceed to the next call.
  Similarly, we proceed to the next call if the size of the produced grammar exceeds 4 times the expected size.
  Otherwise, the grammar is of size $\Oh(\delta \log \frac{n}{\delta})$, so return it to the output.
  
  By Markov's inequality, each call is successful with probability at least $\frac12$.
  Consequently, the number of calls is constant in expectation and $\Oh(\log n)$ with high probability.
  Since each call has a time limit of $\Oh(n)$, this yields the claimed upper bounds on the overall running time.
\end{proof}

\section{Accessing and indexing in $\delta$-bounded space} \label{sec:bt}

Christiansen et al.~\cite[Appendix~A]{CEKNP19} showed how, given a run-length context-free grammar of size $g_{rl}$ generating a string~$S$, we can build a data structure of size $O(g_{rl})$ that supports access and indexed searches on $S$. A direct corollary of their results and Theorem~\ref{thm:grl} is that we can not only represent a string within $O(\delta\log\frac{n}{\delta})$ space, but also support fast access and indexed searches within that space. We can also support the computation of Karp--Rabin signatures \cite{KR87} on arbitrary substrings of $S$.


\begin{corollary}\label{cor:grl}
Given a string $S[1\dd n]$ with measure $\delta$, there exists a data structure of size $O(\delta\log\frac{n}{\delta})$ that can be built in $O(n\log n)$ expected time and (1) can retrieve any substring $S[i\dd i+\ell]$ in time $O(\ell+\log n)$, (2) can compute the Karp--Rabin fingerprint of any substring of $S$ in time $O(\log n)$, and (3) can report the $occ$ occurrences of any pattern $P[1\dd m]$ in $S$ in time $O(m\log n + occ \log^\epsilon n)$, for any constant $\epsilon>0$ fixed at construction time.
\end{corollary}
\begin{proof}
Points (1), (2), and (3) follow from \cref{thm:grl} combined with Theorems A.1, A.3, and A.4 of Christiansen et al.~\cite{CEKNP19}, respectively. Those structures are built in $O(n\log n)$ expected time, which dominates the $O(n)$ expected time needed to build the run-length grammar.
\end{proof}

On the other hand, no known $O(g_{rl})$-space index can efficiently \emph{count} the number of times $P[1\dd m]$ occurs in $S$. Christiansen et al.~\cite[Appendix~A]{CEKNP19} instead show an index of size $O(g)$ that can count in time $O(m\log^{2+\epsilon} n)$ for any fixed constant $\epsilon>0$. We now show that the same can be obtained within space $O(\delta\log\frac{n}{\delta})$. 
Though this space is always $\Omega(g_{rl})$ (Theorem~\ref{thm:rlslp}), it can be $o(g)$ (\cref{thm:gvsdelta}).

Later, we show how the results of Corollary~\ref{cor:grl} can be improved in some cases by using block trees \cite{BGGKOPT15} instead of grammars.
The block tree is a data structure designed to represent repetitive strings $S[1\dd n]$ in $O(z\log\frac{n}{z})$ space while accessing individual symbols of $S$ and computing fingerprints in time $O(\log\frac{n}{z})$, that is, faster than in Corollary~\ref{cor:grl} when $z$ is not too small. We show that the block tree is easily tuned to use the worst-case-optimal $O(\delta\log\frac{n}{\delta})$ space while retaining its access time.

\subsection{Counting}

As explained, it is possible to count how many times $P[1\dd m]$ occurs in $S$ in space proportional to any context-free grammar generating $S$. The idea, however, has not been extended to handle run-length rules of the form $A \rightarrow B^t$. Christiansen et al.~\cite[Section~7]{CEKNP19} accomplished this only for their particular run-length context-free grammar, of size $O(\gamma\log\frac{n}{\gamma})$. We now show that their result can be carried over to our run-length grammar of Section~\ref{sec:ub-gram}.

We start proving a technical point about our grammar, which we will need to establish the result.

\begin{definition}[\cite{CR02}]
A \emph{period} of a string $S[1\dd n]$ is a positive integer $p$ such that $S[p+1\dd n] = S[1\dd n-p]$. We denote the smallest period of $S$ with $\per(S)$.
\end{definition}

\begin{lemma}[{cf.~\cite[Lemma~6.17]{KK19}}]\label{lem:runlen}
For every rule of the form $A \rightarrow B^t$ in our run-length grammar, $\per(\val(A))=|\val(B)|$.
\end{lemma}
\begin{proof}
Observe that $|\val(B)|$ is a period of $\val(A)$ because $\val(A)=\val(B)^t$.
Thus, by the Periodicity Lemma~\cite{fine1965uniqueness}, $\per(\val(A))=\per(\val(B))=\frac1s|\val(B)|$ holds for some integer $s \ge 1$. For a proof by contradiction, suppose that $s>1$.

Let $\ell$ be the minimum level such that $A$ occurs in $S_\ell$,
and let us fix an occurrence of $A$ in~$S_\ell$.
In each string $S_{k}$ with $k\in [0\dd \ell]$, this occurrence corresponds to a fragment $x_k$ satisfying $\val(x_k)=\val(A)$. Note that $x_\ell = A$ and $x_{\ell-1}=B^t$. Consequently, for each $k\in [0\dd \ell)$,
we have $x_k = Y_k^t$ for some string $Y_k\in \Symb^+$; in particular, $Y_{0}=\val(B)$ and $Y_{\ell-1}=B$.

Let us fix the largest $k\in [0\dd \ell)$ such that $Y_k = Z_k^{s}$ for some string $Z_k\in \Symb^+$; note that $k$ is well defined due to $\per(Y_0)=\frac1s|Y_0|$ and $k < \ell-1$ due to $\per(Y_{\ell-1})=|Y_{\ell-1}|=1$.

Let us decompose $x_k$ into $st$ equal-length fragments $x_k = z_{k,1}\cdots z_{k,st}$, and note that $z_{k,i}=Z_k$ for $i\in [1\dd st]$.
Consider the partition of $S_k$ into blocks that are then collapsed to form $S_{k+1}$.
Since $x_{k}$ gets collapsed to $x_{k+1}$, block boundaries are placed before $z_{k,1}$ and after $z_{k,st}$.
Since $z_{k,1}\cdots z_{k,s}=Y_k$ gets collapsed to a fragment matching $Y_{k+1}$, 
a block boundary is also placed between $z_{k,s}$ and $z_{k,s+1}$. 
However, according to \cref{it:rle,it:pair}, block boundaries are placed solely based on the two adjacent symbols,
and therefore a block boundary is placed between every $Z_k[|Z_k|]$ and $Z_k[1]$ and, in particular,
between $z_{k,i}$ and $z_{k,i+1}$ for all $i\in [1\dd st)$.
Consequently, each fragment $z_{k,i}$ is collapsed to a fragment of $S_{k+1}$, which we denote $z_{k+1,i}$. Since $\val(z_{k+1,i})=\val(z_{k,i})=\val(Z_k)$ holds for all $i\in [1\dd st]$, \cref{fct:cons} yields a string $Z_{k+1}\in \Symb^+$ satisfying $z_{k+1,i}=Z_{k+1}$ for all $i\in [1\dd st]$.
This implies $Y_{k+1}=Z_{k+1}^s$, contradicting the choice of $k$. Hence, $s=1$.
\end{proof}

We are now ready to establish the main result.

\begin{theorem}
Given a string $S[1\dd n]$ with measure $\delta$, there exists a data structure of size $O(\delta\log\frac{n}{\delta})$ that can be built in $O(n\log n)$ expected time and can count the number of times any pattern $P[1\dd m]$ occurs in $S$ in time $O(m\log^{2+\epsilon}n)$, for any constant $\epsilon>0$ fixed at construction time.
\end{theorem}
\begin{proof}
We use the technique that Christiansen et al.~\cite[Section~7]{CEKNP19} developed for their particular run-length grammar. The only point where their structure requires some specific property of their grammar is their Lemma 7.2, which we have reproved for our grammar as \cref{lem:runlen}.

The total space is proportional to the size of the grammar, which in our case is $O(\delta\log\frac{n}{\delta})$. Their expected construction time is $O(n\log n)$, which dominates the time to build our run-length grammar. From the time analysis in~\cite{CEKNP19}, it follows that the query time is $O(m\log^{2+\epsilon} n)$ because we split $P$ in $m-1$ places, not just in $O(\log m)$ places as their special grammar allows.
\end{proof}

\subsection{Block trees}

Given integer parameters $\tau$ and $s$, the 
root 
of the block tree divides $S$ into $s$ equal-sized (that is, with the same number of characters) blocks (assume for simplicity that $n = s \cdot \tau^t$ for some integer $t$).%
\footnote{Otherwise, we simply pad $S$ with spurious symbols at the end; whole spurious blocks are not represented. The extra space incurred is only $O(\tau h)$ for a tree of height $h$.}
Blocks are then classified into \emph{marked} and \emph{unmarked}.
If two adjacent blocks $B', B''$ form the leftmost occurrence of the underlying substring $B'\cdot B''$, then both $B'$ and $B''$ are marked.
Blocks $B$ that remain unmarked are replaced by a pointer to the pair of adjacent blocks $B',B''$ that contains the leftmost occurrence of $B$, and the offset $\epsilon \ge 0$ where $B$ starts inside $B'$. Marked blocks are divided into $\tau $ equal-sized sub-blocks, which form the children of the current block tree's level, and processed recursively in the same way.
Let $\sigma$ be the alphabet size. 
The level where the blocks' lengths fall below $\log_\sigma n$ 
contains the leaves of the block tree,
whose blocks store their plain string content using $O(\log n)$ bits. The height of the block tree is then $h = O(\log_\tau \frac{n/s}{\log_\sigma n}) = O(\log_\tau \frac{n\log\sigma}{s\log n}) \subseteq O(\log \frac{n}{s})$.

The block tree construction guarantees that the blocks $B'$ and $B''$ to which any unmarked block points exist and are marked. Therefore, any access to a position $S[i]$ can be carried out in $O(h)$ time, by descending from the root to a leaf and spending $O(1)$ time in each level: To obtain $B[i]$ from a marked block $B$, we simply compute to which sub-block $B[i]$ belongs among the children of $B$.
To obtain $B[i]$ from an unmarked block $B$ pointing to $B',B''$ with offset $\epsilon$, we switch either to $B'[\epsilon+i]$ or to $B''[\epsilon+i-|B'|]$, which are marked blocks.

By storing further data associated with marked and unmarked blocks, the block tree offers the following functionality~\cite{BGGKOPT15}:

\begin{description}
\item[access:] any substring $S[i\dd i+\ell-1]$ is extracted in time $O(h \lceil \ell/\log_\sigma n\rceil)$;
\item[rank:] the number of times symbol $a$ occurs in $S[1\dd i]$, denoted $rank_a(S,i)$, is computed in time $O(h)$ by multiplying the space by $O(\sigma)$;
\item[select:] the position of the $j$th occurrence of symbol $a$ in $S$, denoted $select_a(S,j)$, is computed in time $O(pred(s,n) +h\,pred(\tau,n))$ by multiplying the space by $O(\sigma)$, where $pred(x,n)$ is the time of a predecessor query on a set of $x$ elements from the universe $[1\dd n]$.
\end{description}

It is shown that there are only $O(z\tau )$ blocks in each level of the block tree (except the first, which has $s$ block); therefore the block tree size is $O(s+z\tau \log_\tau \frac{n\log\sigma}{s\log n})$.

\subsubsection{Bounding the space in terms of $\delta$}

We now prove that there are only $O(\delta \tau )$ blocks in each level of the block tree except the root level, and therefore, choosing $s=\delta$ yields a structure of size $O(\delta \tau \log_\tau \frac{n\log\sigma}{\delta\log n})$ with height $O(\log_\tau \frac{n\log\sigma}{\delta\log n})$. For $\tau =O(1)$, the space is $O(\delta\log\frac{n}{\delta})$ and the height is $O(\log \frac{n}{\delta})$.

Let us call level $k$ of the block tree the one where blocks are of length $\tau ^k$ (recall that we assume $n = s \cdot \tau^t)$. 
In level $k$, then, $S$ is covered regularly with blocks $B=S[\tau ^k (i-1)+1\dd \tau ^k i]$ of length $\tau ^k$ (though not all of them are present in the block tree). 
Note that $k$ reaches its maximum in the root (where we have the largest blocks) and the minimum in the leaves of the block tree.

\begin{lemma}\label{lem:deltamarked}
The number of marked blocks of length $\tau ^k$ in the block tree is $O(\delta)$.
\end{lemma}
\begin{proof}
Any marked block $B$ must belong to a sequence of three blocks, $B^- \cdot B \cdot B^+$, such that $B$ is inside the leftmost occurrence of $B^- \cdot B$ or $B \cdot B^+$, or both ($B^-$ and $B^+$ do not exist for the first and last block, respectively).

For the sake of computing our bound, let $\#$ be a symbol not appearing in $S$ and let us add $2\cdot \tau ^k$ characters equal to $\#$ at the beginning of $S$ and $2\tau ^k$ characters equal to $\#$ at the end of $S$. We index the added prefix in negative positions (up to index 0), so that $S[-2\cdot \tau ^k+1 \dd 0] = \#^{2\cdot \tau ^k}$.
Now consider all the $\tau ^k$ text positions $p$ belonging to a marked block $B$.
The \emph{long} substring
$E = S[p-2\cdot \tau ^k\dd p+2\cdot \tau ^k-1]$ centered at $p$, of length $4\tau ^k$, contains $B^- \cdot B \cdot B^+$, and thus $E$ contains the leftmost
occurrence $L$ of $B^- \cdot B$ or $B \cdot B^+$. All those long substrings $E$ must then
be distinct: if two long substrings $E$ and $E'$ are equal, and $E'$ appears after $E$ in $S$, then $E'$ does not contain the leftmost occurrence of any substring $L$. 

Since we added a prefix of length $2\cdot \tau ^k$ and a suffix of length $2\tau ^k$ consisting of character $\#$ to $S$, the number of distinct substrings of length $4\tau ^k$ is at most $d_{4\tau ^k}(S) + 4\tau ^k$.
Therefore, there can be at most $d_{4\tau ^k}(S) + 4\tau ^k$ long substrings $E$ as well, because they must all be distinct. Since each position $p$ inside a block $B$ induces a distinct
long substring $E$, and each marked block
$B$ contributes $\tau ^k$ distinct positions $p$, there are at most $(d_{4\tau ^k}(S)+ 4\tau ^k)/\tau ^k$ marked blocks $B$ of
length $\tau ^k$. The total number of marked blocks of length $\tau ^k$ is thus at most $(d_{4\tau ^k}(S) + 4\tau ^k)/\tau ^k = 4 \cdot d_{4\tau ^k}(S)/(4\tau ^k) + 4\tau ^k / \tau ^k \le 4\delta + 4$.
\end{proof}

Since the block tree has at most $\Oh(\delta)$ marked blocks per level, it has $O(\delta \tau )$ blocks across all the levels except the root level. This yields the following result.

\begin{theorem}\label{thm:delta}
Let $S[1\dd n]$, over alphabet $[1\dd\sigma]$, have measure $\delta$. 
Then the block tree of $S$, with parameters $\tau $ and $s$, is of size $O(s+\delta \tau \log_\tau \frac{n\log\sigma}{s\log n})$ words and height $h = O(\log_\tau \frac{n\log\sigma}{s\log n})$.
\end{theorem}

%

\subsubsection{Operations on block trees}

By properly parameterizing the block tree, we obtain a structure that uses the same asymptotic space and, in some cases, extracts substrings faster than the result of Corollary~\ref{cor:grl}.

\begin{corollary}
Let $S[1\dd n]$, over alphabet $[1\dd\sigma]$, have measure $\delta$. 
Then a block tree of $S$ can use $O(\delta \log \frac{n}{\delta})$ space and extract a substring of length $\ell$ from $S$ in time $O(\lceil \ell/\log_\sigma n\rceil \log\frac{n}{\delta})$.
\end{corollary}
\begin{proof}
We obtain the desired space $O(\delta\log\frac{n\log\sigma}{\delta\log n}) \subseteq O(\delta\log\frac{n}{\delta})$ by using Theorem~\ref{thm:delta} with $s=\delta$ and $\tau=O(1)$.
The height is $O(\log\frac{n}{\delta})$, and thus the substring extraction costs $O(\lceil \ell/\log_\sigma n\rceil \log\frac{n}{\delta})$.
\end{proof}

Navarro and Prezza \cite{NP18} show how the Karp--Rabin signature of any $S[i\dd j]$ can be computed in time $O(\log\frac{n}{\gamma})$ on their $\Gamma$-tree variant of the block tree, which is of size $O(\gamma\log\frac{n}{\gamma})$. We now extend their result so as to compute the {\em fingerprint} $\phi(S[i\dd j]) = \phi(S[i])\, \circ\, \phi(S[i+1]) \circ \cdots \circ \,
\phi(S[j])$ for any group operator $\circ$ within $O(\delta\log\frac{n}{\delta})$ space and using $O(\log\frac{n}{\delta})$ group operations, by enhancing the original block trees. This includes computing the Karp--Rabin signature of $S[i\dd j]$ in time $O(\log\frac{n}{\delta})$, because all the required group operations can be supported in constant time on those signatures \cite{NP18}.

\begin{definition}
Let $\Sigma$ be an alphabet and $(\mathbb{G},\circ,^{-1},0)$ a group. A function $\phi : \Sigma^* \rightarrow \mathbb{G}$ is a {\em fingerprint} on $\Sigma^*$ if $\phi(\varepsilon)=0$ and $\phi(S \cdot a) = \phi(S) \circ \phi(a)$ for every $S \in \Sigma^*$ and $a\in\Sigma$.
\end{definition}

\begin{theorem}\label{thm:KR}
Let $S[1\dd n]$ be a string on alphabet $\Sigma$, $(\mathbb{G},\circ,^{-1},0)$ a group where $\log |\mathbb{G}| = O(\log n)$, and $\phi:\Sigma^* \rightarrow \mathbb{G}$ a fingerprint on $\Sigma^*$. Then there is a data structure of size $O(\delta\log\frac{n}{\delta})$ that can compute any $\phi(S[i\dd j])$ in time $O(\log\frac{n}{\delta})$.
\end{theorem}
\begin{proof}
We use a block tree for $S$ where the leaves handle only one symbol, and we augment it as follows. Together with every stored block $B=S[\tau^k(i-1)+1\dd \tau^k i]$ at every level $k$, we store its fingerprint $\phi(B)$ (using constant space). Furthermore, at the top level, say level $k=\kappa$, we store the fingerprint $\phi(S[1\dd\tau^\kappa (i-1)])$ at the block-tree node corresponding to the block $B=S[\tau^\kappa(i-1)+1\dd \tau^\kappa i]$. In addition, let $B_1,\ldots,B_\tau$ be the children at level $k-1$ of a marked block $B$ of level $k$. We store, at each such child $B_j$, the fingerprint $\phi(B_1\cdots B_{j-1})$. 

We will compute any fingerprint $\phi(S[i\dd j])$ as $\phi(S[1\dd i-1])^{-1} \circ \phi(S[1\dd j])$; therefore we focus on computing only fingerprints of the form $\phi(S[1\dd i])$ for arbitrary $i$.
%
At the top level $\kappa$, the prefix $S[1\dd i]$ spans a sequence $B_1\cdots B_t$ of blocks followed by a (possibly empty) prefix $C$ of block $B_{t+1}$. Since $\phi(B_1\cdots B_t)$ is explicitly stored at block $B_{t+1}$, the problem reduces to computing $\phi(C)$ and then returning $\phi(B_1\cdots B_t) \circ \phi(C)$. The following is needed only if $C \neq \varepsilon$.

To compute the fingerprint of a prefix $B[1\dd l]$ of an explicit block $B$ at level $k \le \kappa$ (so $1 \le l \le \tau^k$), we distinguish two cases.

\begin{enumerate}
\item $B$ is a marked block, with children $B_1,\ldots,B_\tau$ at level $k-1$, so that $B[l]$ belongs to the child $B_j$ (i.e., $j=\lceil l/\tau^{k-1}\rceil$). We then return $\phi(B_1 \cdots B_{j-1}) \circ \phi(B_j[1\dd l \!\!\mod \tau^{k-1}])$, where the first term is stored at $B_j$ and the second is computed from the next level (only necessary if $l\!\!\mod \tau^{k-1} \neq 0$).

\item $B$ is an unmarked block, pointing to a previous occurrence inside $B' \cdot B''$ at the same level $k$, with both $B'$ and $B''$ marked. If the occurrence of $B$ spans only one marked block, $B'$, then we replace $B$ by $B'$ in our query and we are back in case (1). Otherwise, let $B[1\dd \tau^k] = B'[i\dd \tau^k]\cdot B''[1\dd i-1]$. For each pointer of this kind in the block tree, we store the fingerprint $\phi(B'[i\dd \tau^k])$ at $B$. We consider two sub-cases. 

\begin{enumerate}
\item If $l\geq \tau^k-i+1$, then $B[1\dd l] = B'[i\dd\tau^k]\cdot B''[1\dd l-(\tau^k-i+1)]$. We then return $\phi(B'[i\dd \tau^k]) \circ \phi(B''[1\dd l-(\tau^k-i+1)])$, where the first term is stored at $B$ and the second is a new prefix problem on level $k$, but now of case (1).

\item If $l < \tau^k-i+1$, then $B[1\dd l] = B'[i\dd i+l-1]$. Although this
is neither a prefix nor a suffix of a block, note that $B[1\dd l]\cdot B'[i+l\dd \tau^k] = B'[i\dd i+l-1] \cdot B'[i+l\dd \tau^k] = B'[i\dd \tau^k]$. We then return $\phi(B'[i\dd \tau^k]) \circ \allowbreak \phi(B'[i+l\dd \tau^k])^{-1}$. The first fingerprint is stored at $B$, whereas the second is the fingerprint of the suffix of the marked block $B'$. We compute it as $\phi(B'[i+l\dd \tau^k]) = \phi(B'[1\dd i+l-1])^{-1} \circ \phi(B')$, where $\phi(B')$ is stored at $B'$ and the first term is again the fingerprint of a prefix at the same level, but now of case (1).
\end{enumerate}
\end{enumerate}

To sum up, computing a prefix of an explicit block at level $k$
reduces to the problem of computing a prefix of an explicit block at level
$k-1$ plus a constant amount of group operations to combine values. In
the worst case, we navigate down to the leaves, where fingerprints of single characters can be computed for free. Since the height of this block tree is $\log\frac{n}{\delta}$, the whole fingerprint is computed at cost $O(\log\frac{n}{\delta})$, counting group operations and other costs.
\end{proof}

We note that it is unknown if a result like Theorem~\ref{thm:KR} can be obtained on a semigroup. For example, it is not known how to compute the minimum of a substring in polylogarithmic time on block trees \cite{CN19,BGGKOPT15}.

Though our times for accessing and fingerprinting seem to be larger than those obtained on other block tree variants, which are of height $O(\log\frac{n}{z})$ \cite{BGGKOPT15} or $O(\log\frac{n}{\gamma})$ \cite{NP18}, we next show that $\log\frac{n}{\delta}$ is asymptotically equal to $\log\frac{n}{g}$, which also encompasses all the intermediate measures, $\delta \le \gamma \le b \le c \le z \le g_{rl} \le g$.

\begin{lemma}
Let $x = O(\delta\log^c \frac{n}{\delta})$ for some constant $c>0$. Then
$\log\frac{n}{\delta} = O(\log\frac{n}{x})$. As a consequence, $\log\frac{n}{\delta}=\Theta(\log\frac{n}{g})$.
\end{lemma}
\begin{proof}
From the hypohesis it follows that
$\frac{n}{\delta} = O(\frac{n}{x}\log^c \frac{n}{\delta})$. Since $\log^c\frac{n}{\delta}=O(\sqrt{\frac{n}{\delta}})$, it holds that $\sqrt{\frac{n}{\delta}}=O(\frac{n}{x})$,
and thus $\log \frac{n}{\delta} = O(\log \frac{n}{x})$.
The final consequence follows from $\delta \le g = O(z\log\frac{n}{z}) = O(\delta\log\frac{n}{\delta}\log\frac{n}{z}) = O(\delta\log^2\frac{n}{\delta})$ \cite{Ryt03,CLLPPSS05,RRRS13}.
\end{proof}

Hence, the times we obtain using $O(\delta \log \frac{n}{\delta})$ space, not only for access but also for rank and select, and for computing fingerprints, are asymptotically the same as those obtained in $O(\gamma\log \frac{n}{\gamma})$ space~\cite{NP18,prezza2019optimal}
or in $O(z \log \frac{n}{z})$ space~\cite{BGGKOPT15}.

Finally, it is also possible to obtain the same result as point (3) of Corollary~\ref{cor:grl} using block trees; see the conference version of this article \cite{KNP20}.



\section{Conclusions}

We have made a step towards establishing the right measure of repetitiveness for a string $S[1\dd n]$. Compared with the most principled prior measure, the size $\gamma$ of the smallest attractor, the proposed measure $\delta$ has several important advantages:
\begin{enumerate}
\item It can be computed in linear time, while finding $\gamma$ is NP-hard. It is also insensitive to simple string transformations (reversals, alphabet permutations) and, unlike $\gamma$, monotone with respect to appending symbols.
\item It lower bounds the previous measure, $\delta \le \gamma$, with up to a logarithmic-factor separation. For every $n$ and $2 \le \delta \le n$, there are string families where $\gamma = \Omega(\delta\log\frac{n}{\delta})$.
\item We can always encode $S$ in $O(\delta\log\frac{n}{\delta})$ space, and this is
worst-case optimal in terms of $\delta$: for any length $n$ and any value $2\le \delta \le n$, there are string families needing $\Omega(\delta\log\frac{n}{\delta})$ space. Thus, $o(\delta\log n)$ space is unreachable in general. Instead,
no string family is known to require $\omega(\gamma)$ space, nor it is known if $o(\gamma\log n)$ space can always be reached.
\item We can build a run-length context-free grammar of size $O(\delta\log\frac{n}{\delta})$, which then upper bounds the size $g_{rl}$ of the smallest such grammar, and transitively the measures $\gamma$, $b$, $c$, $v$, and $z$. At the same time, there are string families where the smallest
context-free grammar is of size $g=\Omega(\delta \log^2 \frac{n}{\delta} / \log\log \frac{n}{\delta})$. No such separation is known for $\gamma$.
\item There are encodings using $O(\delta\log\frac{n}{\delta})$ space and supporting direct access and indexed searches, with the same complexities obtained within attractor-bounded space, $O(\gamma\log\frac{n}{\gamma})$~\cite{NP18}. An exception is a very recent faster index~\cite{CEKNP19}.
\end{enumerate}

An ideal compressibility measure capturing repetitiveness should be reachable, monotone, invariant to simple string transformations, efficient to compute, and optimal within a hopefully refined partition of the strings. The measure $\delta\log\frac{n}{\delta}$ is reachable, monotone, invariant, fast to compute, and optimal within the class of all the strings with the same $n$ and $\delta$ values.

In comparison, measure $b$ (the size of the smallest bidirectional macro scheme) is reachable and invariant, but it is non-monotonic and NP-hard to compute. On the other hand, it is optimal within a more refined partition, since it is always $O(\delta\log\frac{n}{\delta})$. The size $\gamma \le b$ of the smallest attractor is unknown to be reachable, and it is non-monotone and NP-hard to compute, yet invariant. If it turns out that one can always encode a string within $O(\gamma)$ space, then $\gamma$ would be a reachable measure even more refined than $b$.

The measure $\delta\log\frac{n}{\delta}$ is then a good candidate in the fascinating quest for an ideal measure of repetitiveness. Its main weakness is that it is optimal within a partition of the strings that, though reasonably refined, is improved by other measures (which have other weaknesses).

On the more algorithmic side, it would be useful to
compute $\delta$ within little space. Its $O(n)$-time computation~\cite{CEKNP19} requires also $O(n)$ space, which can be unaffordable for very large text collections. Bernardini et al.~\cite{BFGP20}, for example, show how to compute it in time $O(n^3/s^2)$ and $O(s)$ space. It would also be worth obtaining faster indexes of size $O(\delta\log\frac{n}{\delta})$. Our index requires $O(m\log n + occ \log^\epsilon n)$ search and $O(m\log^{2+\epsilon} n)$ count time,
while in $O(\gamma\log\frac{n}{\gamma})$ space it is possible to search in $O(m+(occ+1)\log^\epsilon n)$ and count in $O(m+\log^{2+\epsilon}n)$ time~\cite{CEKNP19}.

\section*{Acknowledgements}
Part of this work was carried out during the Dagstuhl Seminar 19241, ``25
Years of the Burrows-Wheeler Transform''.
We also thank Travis Gagie for pointing us the early reference related to $\delta$~\cite{RRRS13}.

TK was supported by ISF grants no. 1278/16, 824/17, and 1926/19, a BSF grant no. 2018364, and an ERC grant MPM (no. 683064) under the EU's Horizon 2020 Research and Innovation Programme. GN was supported in part by ANID -- Millennium Science Initiative Program -- Code
ICN17\_002, and Fondecyt Grant 1-200038, Chile.

\bibliographystyle{IEEEtran}
\bibliography{paper}

\appendix
\section{Proof from~\cite{IPM}}\label{app:proofs}
In this Appendix, we reproduce proofs from an unpublished manuscript~\cite{IPM}.

\fctcons*
\begin{proof}
    We proceed by induction on $k$. 
    Let $x,x'$ be fragments of $S_{k}$ satisfying $\val(x)=\val(x')$.
    If $k=0$, then $x= x'$ holds due to $\val(x)=x$ and $\val(x')=x'$.
    Otherwise, let $\bar{x}$ and $\bar{x}'$ be the fragments of $S_{k-1}$ obtained from $x$ and $x'$, respectively, by expanding collapsed blocks.
    Note that $\val(\bar{x})=\val(x)=\val(x')=\val(\bar{x}')$, so the inductive assumption guarantees
    $\bar{x}= \bar{x}'$.
    Inspecting \cref{it:rle,it:pair}, observe that if $S_{k-1}[i]=S_{k-1}[i']$ and $S_{k-1}[i+1]=S_{k-1}[i'+1]$, 
    then block boundaries after positions $i,i'\in [1\dd |S_{k-1}|)$ are placed consistently:
    either after both of them or after neither of them.
    Consequently, block boundaries within $\bar{x}$ and $\bar{x}'$ are placed consistently.
    Moreover, both $\bar{x}$ and $\bar{x}'$ consist of full blocks (since they are collapsed to $x$ and $x'$, respectively). Thus, $\bar{x}$ and $\bar{x}'$ are consistently partitioned into full blocks. 
    Matching blocks get collapsed to matching symbols both in \cref{it:rle,it:pair}, so we derive $x= x'$.
\end{proof}

\cordistinct*
\begin{proof}
    For a proof by contradiction, suppose that $S_k[j]=S_k[j+1]\in  \Symb_{k+1}$ holds for some $j\in [1\dd |S_k|)$.
    Let $x=S_{k-1}(i-\ell\dd i]$ and $x'=S_{k-1}(i\dd i+\ell']$ be blocks of $S_{k-1}$ collapsed to $S_k[j]$ and $S_k[j+1]$, respectively. Due to $\val(x)=\val(S_k[j])=\val(S_k[j+1])=\val(x')$, \cref{fct:cons} guarantees $x = x'$ and, in particular, $\ell=\ell'$.
    If $\ell=1$, then $S_{k-1}[i]=S_k[j]=S_k[j+1]=S_{k-1}[i+1] \in \Symb_{k+1}$.  
    Otherwise, $x= x' = A^\ell$ for some symbol $A\in \Symb_{k}$,
    which means that $S_{k-1}[i]=S_{k-1}[i+1]=A$.
    In either case, $S_{k-1}[i]=S_{k-1}[i]\in \Symb_{k}=\Symb_{k+1}$, which means that $\rle_{\Symb_k}(S_{k-1})$ does not place a block
    boundary after position $i$ in $S_{k-1}$. This contradicts the choice of $i$
    as the boundary between blocks $x$ and~$x'$.
  \end{proof}

\lemrecompr*
  \begin{proof}
    We proceed by induction on $k$, with a weaker assumption $\alpha \ge 15\ell_k$ for odd $k$.
    In the base case of $k=0$, the claim is trivial due to $B_k = [1\dd n)$.
    Next, we shall prove that the claim holds for integers $k>0$ and $\alpha > \ell_k$
    assuming that it holds for $k-1$ and $\alpha - \floor{\ell_k}$.
    This is sufficient for the inductive step: If $\alpha \ge 16\ell_k$ for even $k$, then $\alpha - \floor{\ell_k} \ge 15\ell_k = 15\ell_{k-1}$. Similarly, if $\alpha\ge 15\ell_k$ for odd $k$, then $\alpha - \floor{\ell_k}  \ge 14\ell_k = 16\ell_{k-1}$.
    
    For a proof by contradiction, suppose that $S(i-\alpha\dd i+\alpha]= S(i'-\alpha\dd i'+\alpha]$ yet $i\in B_{k}$ and $i'\notin B_{k}$ for some $i,i'\in [\alpha\dd n-\alpha]$.
    By the inductive assumption (applied to positions $i,i'$), $i\in B_{k} \sub B_{k-1}$ implies $i'\in B_{k-1}$.
    Let us set $j,j'$ so that  $i = |\val(S_{k-1}[1\dd j])|$ and $i' = |\val(S_{k-1}[1\dd j'])|$.
    Since a block boundary was not placed between $S_{k-1}[j']$ and $S_{k-1}[j'+1]$,
    we have $S_{k-1}[j'],S_{k-1}[j'+1]\in \Symb_k$ (see \cref{it:rle,it:pair}). 
    Consequently, the phrases $S(i'-\ell\dd i']= \val(S_{k-1}[j'])$ and $S(i'\dd i'+r]= \val(S_{k-1}[j'+1])$ around position $i'$ are of length at most $\floor{\ell_k}$. 
    By the inductive assumption (applied to positions $i+\delta,i'+\delta$ for $\delta \in [-\ell\dd r]\sub [-\floor{\ell_k}\dd \floor{\ell_k}]$),
    there are matching phrases $S(i-\ell\dd i]$ and $S(i\dd i+r]$ around position $i$.
    Due to \cref{fct:cons}, this yields $S_{k-1}[j]=S_{k-1}[j']$ and $S_{k-1}[j+1]=S_{k-1}[j'+1]$.
    Consequently, a block boundary was not placed between $S_{k-1}[j]$ and $S_{k-1}[j+1]$,
    which contradicts $i\in B_{k}$.
    \end{proof}

  \lemrand*
    \begin{proof}
        We proceed by induction on $k$. For $k=0$, we have $|S_0|=n < 1+4n = 1+\frac{4n}{\ell_1}$.
        If $k$ is odd, we note that $|S_{k}|\le |S_{k-1}|$ and therefore $\Exp[|S_{k}|] \le \Exp[|S_{k-1}|]
        < 1+\frac{4n}{\ell_{k}} = 1+\frac{4n}{\ell_{k+1}}$.
        Thus, it remains to consider even integers $k>0$.
      
        \begin{claim}\label{clm:rand0}
      For every even $k>0$, conditioned on any fixed $S_{k-1}$, we have $\Exp\left[|S_{k}| \bigm| S_{k-1}\right]< \tfrac14+\tfrac{n}{2\ell_k}+\tfrac34|S_{k-1}|$.
        \end{claim}
        \begin{proof}
          Let us define $J = 
          \{j\in [1\dd |S_{k-1}|] :j=|S_{k-1}|\text{ or }S_{k-1}[j]\notin\Symb_k\text{ or }S_{k-1}[j+1]\notin\Symb_k\}$.
        Since $A\notin \Symb_k$ yields $|\val(A)|>\ell_k$, we have $|J| < 1+\frac{2n}{\ell_k}$.
        Moreover, observe that if $j\in [1\dd |S_{k-1}|]\sm J$, then $S_{k-1}[j]$ and $S_{k-1}[j+1]$ are,
        by \cref{cor:distinct}, distinct symbols in $\Symb_k$. Consequently, $\Pr[S_{k-1}[j]\in \Left_k \text{ and } S_{k-1}[j+1]\in \Right_k] = \tfrac14$.
        Thus, the probability that $\pc_{\Left_k,\Right_k}(S_{k-1})$ places a block boundary after position $j\in [1\dd |S_{k-1}|]\sm J$ is $\frac34$. Therefore,
        \[\Exp\left[|S_{k}| \bigm| S_{k-1}\right] = |J| + \tfrac34(|S_{k-1}|-|J|) = \tfrac14|J|+\tfrac34|S_{k-1}|<\tfrac14+\tfrac{n}{2\ell_k}+\tfrac34|S_{k-1}|.\qedhere\]
        \end{proof}
      
        Since the partition $\Symb_k=\Left_k\cup\Right_k$ is independent of $S_{k-1}$,
        \cref{clm:rand} and the inductive assumption yield
        \[\Exp[|S_k|] < \tfrac14 + \tfrac{n}{2\ell_k}+\tfrac34\Exp[|S_{k-1}|] < \tfrac14+\tfrac{n}{2\ell_k} + \tfrac34+\tfrac{3n}{\ell_k}=1+\tfrac{7n}{2\ell_k}=1+\tfrac{4n}{\ell_{k+1}}.\qedhere\]
      \end{proof}

\end{document}